\documentclass[11pt,a4paper]{article}
\usepackage[utf8]{inputenc}
\usepackage[headheight=15pt,left=3cm,right=3cm,top=3cm,bottom=2.5cm]{geometry}
\usepackage{apacite}
\usepackage{amsmath}
\usepackage{amssymb,bm}
\usepackage{amsthm}
\usepackage{natbib}
\usepackage{multirow}
\usepackage{tabularx}
\usepackage{booktabs}
\usepackage{graphicx}
\usepackage{etoolbox}
\usepackage{placeins}
\usepackage{enumerate}
\usepackage{color}
\usepackage{fancyvrb}
\usepackage{bbm}
\usepackage[ruled]{algorithm2e}
\usepackage{hyperref}
\usepackage[capitalise]{cleveref}

\usepackage{dcolumn}
\usepackage{mathabx}
\usepackage{comment}
\newcolumntype{d}[1]{D{.}{.}{#1}}

\newtheorem{lemma}{Lemma}

\newtheorem{proposition}{Proposition}
\renewcommand{\top}{{\mkern-1.5mu\mathsf{T}}}

\newcommand{\E}{\mathrm{E}}
\newcommand{\Var}{\mathrm{Var}}

\newcommand{\Corr}{\mathrm{Corr}}

\usepackage{authblk}
\usepackage{color,soul}
\SetKw{Continue}{continue}
\defcitealias{UKHLS19}{University of Essex, 2019}

\title{Modelling Correlation Matrices in Multivariate Dyadic Data:
Latent Variable Models for Intergenerational Exchanges of Family Support}
\author[1]{Siliang Zhang}
\affil[1]{Key Laboratory of Advanced Theory and Application in Statistics and Data Science-MOE, School of Statistics, East China Normal University}

\author[2]{Jouni Kuha}
\author[2]{Fiona Steele}
\affil[2]{Department of Statistics, London School of Economics \& Political Science}

\begin{document}
\maketitle

\begin{abstract}
We define a model for the joint distribution of multiple continuous
latent variables which includes a model for how their correlations
depend on explanatory variables. This is motivated by and applied to
social scientific research questions in the analysis of
intergenerational help and support within families, where the
correlations describe reciprocity of help between generations and
complementarity of different kinds of help. We propose an MCMC procedure
for estimating the model which maintains the positive definiteness of
the implied correlation matrices, and describe theoretical results which
justify this approach and facilitate efficient implementation of it. The
model is applied to data from the UK Household Longitudinal Study to
analyse exchanges of practical and financial support between adult
individuals and their non-coresident parents.
\end{abstract}

\vspace*{2ex}

\section{Introduction} 
\label{sec:introduction}

In contemporary low-mortality countries, population ageing has led to an
increase in the need for help and support for people with age-related
functional limitations. At the same time, the need for support may also
be increasing among younger people as a result of delayed transitions to
adulthood, unstable employment, high cost of living, and rises in
divorce and re-partnership rates \citep{lesthaeghe:14,henretta.etal:18}.
With limited public resources available to meet these demands, there is
a greater reliance on private transfers of support within families,
especially between parents and their adult children.  The main
`currencies' of such intergenerational exchanges are time (or practical
support) and money \citep[e.g.][]{grundy:05}. Another form of
intergenerational support is coresidence but its overall rate remains
low, in spite of a small increase in coresidence between young adults
and their parents \citep[e.g.][]{stone.etal:11}. Transfers of practical
and financial support between relatives living in different households
are thus a more important component of family exchanges. Understanding
the nature of these exchanges is important for anticipating which
population sub-groups may be at risk of unmet need for support or
experience a reduced capacity to provide support due to changes in their
circumstances.

Previous research highlights the importance of reciprocity (or symmetry)
in such exchanges, either contemporaneously or over the life course
\citep{hogan.etal:93,grundy:05,silverstein.etal:02}, both as a
motivating factor for providing support and because of its association
with other outcomes. For example, there is evidence that overbenefitting
(receiving more than giving) has negative consequences for older
parents' well-being \citep{davey.eggebeen:98} while balanced exchanges
are positively associated with parents' mental health \citep{litwin:04}.
The extent of reciprocity is likely to depend on individual
characteristics. For example, in a cross-national European study,
\cite{mudrazija:16} finds that net transfers from parents to adult
children follow a similar age pattern across the majority of countries,
with declining positive transfers (parents giving more than they
receive) for parents aged 50-79, becoming negative in most countries
from age 80. There is also evidence from Europe \citep{mudrazija:16} and
the U.S. \citep{hogan.etal:93} that reciprocity reflects the
geographical proximity of parents and children and gender differences in
family roles.

Another question of interest is whether practical and financial support
serve as functional substitutes or complements of each other
\citep[e.g.][]{mudrazija:16}, and how their interdependence depends on
individual characteristics. Among the factors that may play a role are
income and geographical distance where better-off adult children or
children living at a greater distance from their parents may substitute
money for time transfers to parents \citep[e.g.][]{grundy:05}.
Alternatively time and money transfers may be positively associated,
with a tendency to give or receive both or neither form of support.

Most previous substantive research has either combined practical and
financial support, or analysed them separately. The first of these
approaches implicitly assumes that the two forms of support are
indicators of a common underlying dimension and thus does not allow for
differences in their predictors, while the second ignores their
interrelationship. Moreover, most earlier work has analysed support
given separately from support received, which precludes the analysis of
reciprocity of exchanges.  Research that has investigated reciprocity
has typically modelled a joint categorical outcome for whether exchanges
are mutual or one-way \citep[e.g.][]{hogan.etal:93} or modelled the
difference between support given and support received
\citep[e.g.][]{mudrazija:16}. Both approaches consider reciprocity
directly, but then do not permit analysis of the effects of individual
characteristics on exchanges in each direction separately.
Alternatively, reciprocity can be quantified as the residual correlation
in a joint model of support given and support received
\citep{kuhaetal22,steeleetal22}. This approach can be extended to treat
financial and practical support as separate but correlated outcomes
\citep{steeleetal22}.

In this paper we develop a general joint modelling framework that is
used to simultaneously investigate predictors of financial and practical
support given and received, and predictors of the correlations among
these different types of exchange. We analyse cross-sectional
multivariate dyadic data from the UK Household Longitudinal Study
(UKHLS), which contains 16 questions (`items') about exchanges of help
on dyads formed of a survey respondent and their non-coresident
parent(s). Seven of the items relate to whether or not different kinds
of practical help are given to parents (for example, assistance with
shopping) and a further seven items indicate the forms of practical help
received from parents.  The remaining two items indicate whether
financial help is given and received.  The practical help items are
treated as multiple binary indicators of two continuous latent variables
which are modelled jointly with latent variables taken to underlie the
two indicators of financial exchanges. We also account for zero
inflation, which arises from a high proportion of respondents who report
giving or receiving none of the types of support, by including in the
model the joint distribution of two binary latent variables for the
subpopulations with excess zeros.

The model formulation builds on that of \cite{kuhaetal22}, who also
analysed family exchanges of support using UKHLS data. We extend their
model in two ways. First, we separate practical and financial help,
which they considered together. Second, and most importantly, we
introduce a model for how the correlations of tendencies to give and
receive different types of support depend on predictors (covariates).
The key advantage of this framework is that it allows us to answer
questions not only about the predictors of giving and receiving
different forms of support (the mean structure) but also about the
predictors of their correlation structure. The latter is of particular
interest here because it provides information about the symmetry of
exchanges (correlations between giving and receiving help) and
complementarity of different forms of help (correlations between giving
or receiving financial and practical help) for different population
sub-groups. The framework presented here also extends that of
\cite{steeleetal22} who separated practical and financial support in a
joint longitudinal model of bidirectional exchanges using UKHLS data,
but with the seven items for practical help given and received collapsed
into two binary outcomes, and without predictors for the correlation
structure.

Methodologically, this paper contributes to the literature on modelling
correlation or covariance matrices given covariates. We review this
literature in Section~\ref{sec:model_corr}. A key challenge here is that
the estimated matrices should be positive definite at all relevant
values of the covariates (for some approaches the constraint that the
diagonal elements of a correlation matrix should be 1 is also
challenging, so they model the covariance matrix instead). Broadly, two
approaches may be taken to deal with this
\citep{pinheiro1996unconstrained}. `Unconstrained' methods specify a
model for some transformation which ensures that the fitted matrix will
be positive definite, while `constrained' methods enforce this condition
during estimation. A disadvantage of the unconstrained approach is that
the parameters of the transformation are often not easily interpretable.
Constrained estimation, in contrast, can use interpretable models for
the covariances or correlations themselves, but it faces the challenge
of how to actually implement the constraint.

We adopt a constrained approach of estimation. We first decompose the
covariance matrix into the standard deviations and the correlation
matrix, and specify a linear model for each correlation given covariates
(in our application we do not use covariates in the model for the
standard deviations, but they could easily be included). The estimation
is carried out in the Bayesian framework, using a tailored MCMC
algorithm. Here the constraint is implemented by checking it at each
MCMC sampling step, so that the most recently sampled parameters can
only be retained if they imply a positive definite correlation matrix at
all relevant values of the covariates. This builds on methods proposed
previously without covariates \citep{johnbar2000,
wongEfficientEstimationCovariance2003}, which we extend to models that
include individual-level covariates for the correlations. Theoretical
results on properties of correlation matrices which justify this
approach and an efficient implementation of it are described in Section
\ref{sub:Some Useful Properties}.

In the rest of the paper, the UKHLS data are introduced in Section
\ref{sec:data_on_help_exchanges}, and the specification of the joint
model is described in Section \ref{sec:latent_variable_model}. Section
\ref{sec:model_corr} reviews previous literature on modelling covariance
and correlation matrices, and Section \ref{sub:Some Useful Properties}
and \ref{app:proofs_of_lemmas} give the theoretical results that provide
the basis of our estimation of the model for the correlations.
Estimation of the joint model is described in Section
\ref{sec:estimation_of_the_latent_variable_model} and \ref{app:mcmc}.
Results of the  analysis of intergenerational exchanges of family
support are then described in
Section~\ref{sec:analysis_of_the_help_exchanges}, and a concluding
discussion is given in Section \ref{sec:conclusion}.

\section{Data} 
\label{sec:data_on_help_exchanges}

We use data from the Understanding Society survey, also known as the UK
Household Longitudinal Study (UKHLS) \citepalias{UKHLS19}. This is a
long-standing household panel survey. Our analysis does not make use of
its longitudinal features, but carries out a cross-sectional analysis of
data from one wave of the survey, collected in 2017--19 (wave 9 of
UKHLS). This included the `family network' module which collected
information on exchanges of help with relatives living outside a
respondent's household.

Respondents who had at least one non-coresident parent were asked
whether they `nowadays' `regularly or frequently' gave each of eight
types of help to their parent(s): lifts in a car; help with shopping;
providing or cooking meals; help with basic personal needs; washing,
ironing or cleaning; personal affairs such as paying bills or writing
letters; decorating, gardening or house repairs; or financial help.
These items are dichotomous, with the response options `Yes' and `No'.
The same questions were asked about receipt of support from parents, but
with `personal needs' replaced by `help with childcare'. In the analysis
that follows we will distinguish between financial help (measured by a
single item for support in each direction) and practical help (measured
by the remaining seven items). Where a respondent had both biological
and step/adoptive parents alive, the respondents were asked to report on
the ones that they had most contact with. Although respondents were
asked about giving parents a lift in their car `if they have one', the
recorded variable had no missing values for this item. We therefore used
other survey information to set this item to missing for respondents who
did not have access to a car. Similarly, the childcare item was coded as
missing for respondents who did not have coresident dependent children
aged 16 or under. For the item on receiving lifts from parents, we do
not have information on whether the parents have access to a car, so
responses of `No' to this item will include also cases where they do
not.

We consider as covariates a range of individual and household
demographic and socioeconomic characteristics that aim to capture an
adult child's capacity to give help to their parents and their potential
need for support from their parents. Most variables in the survey refer
to the respondent (the child in our case), as less information was
collected on non-coresident relatives, but we also include a small set
of characteristics of the parents that are indicative of their need and
capacity for support similarly. The following respondent characteristics
were included: age, gender, whether they have a coresident partner,
indicators of the presence and age of their youngest biological or
adopted coresident children, the number of siblings (as a measure of
both alternative sources of support for parents and competition for the
receipt of parental support), whether they have a long-term illness that
limits their daily activities, employment status (classified as employed
or non-employed [unemployed or economically inactive]), education (up to
secondary school only, or post-secondary qualifications), household
tenure (home-owner or social/private renter), and household income
(equivalised, adjusted for inflation using the 2019 Consumer Price
Index, and log transformed). The parental characteristics included were
the age of the oldest living parent and whether either parent lives
alone.  We also include the travel time to the nearest parent,
dichotomized as 1 hour or less vs.\ more than 1 hour.

The analysis sample was first restricted to the 15,825 respondents aged
18 or over who had at least one non-coresident parent but no coresident
parent. We excluded respondents whose nearest parent lived or worked
abroad (1830 of them), because the nature of their exchanges is likely
to differ from parents based in the UK, and then also omitted 1792
respondents who had missing data on any covariate or on all the help
items.  The final sample size for analysis is $n=12,203$. Because of the
design of UKHLS, the sample can include some respondents who are
siblings to each other. However, preliminary analysis indicated that
their number was very small for our analysis sample, so we ignore this
feature and treat all the respondents as independent of each other.
Table \ref{t_items} shows the percentages of positive response on each
of the help items for the analysis sample, and \cref{t_covariates} shows
descriptive statistics for the covariates.

\begin{table}[!htbp]
\centering
\caption{Percentage of respondents who reported giving help to their parents and
receiving help from the parents, by item.}
\label{t_items}
\begin{tabular}{lrr}
\hline
 & Help given & Help received \\
Item & to parents & from parents \\
\hline
Practical help (7 items):                                &       & \\
\hspace*{1em}Lifts in car                            & 31.1  & 11.8 \\
\hspace*{1em}Shopping                              & 21.5  & 7.3 \\
\hspace*{1em}Providing or cooking meals              & 11.9  & 12.0 \\
\hspace*{1em}Basic personal needs (to parents only)            & 3.7   & -- \\
\hspace*{1em}Looking after children (from parents only)              & --    & 39.0 \\
\hspace*{1em}Washing, ironing or cleaning        & 7.1   & 4.9 \\
\hspace*{1em}Personal affairs             & 17.1  & 2.0 \\
\hspace*{1em}Decorating, gardening or house repairs     & 17.7  & 7.3 \\[.5ex]
Financial help (1 item)                                 & 6.0   & 12.7
\\[2ex]
\emph{At least one of the seven kinds of practical help:} &43.2 & 33.4 \\
\emph{At least one of any kind of practical or financial help:} &44.4 & 38.2 \\
\hline
\multicolumn{3}{p{1.0\textwidth}}{\footnotesize{Data from UKHLS, 2017-19 (Wave 9).
The overall sample size is $n=12,203$. The percentages for the
individual items are based on observed samples for them, excluding cases with missing data.
The item on giving lifts to parents is missing for the 17.1\% of respondents who have no access to a car, and
the item on childcare is missing for the 54.1\% respondents who have no co-resident dependent children.}}
\end{tabular}
\end{table}

\begin{table}[!htbp]
\centering
\caption{Descriptive statistics of the covariates used in the analysis.}
\label{t_covariates}

\vspace*{1ex}
\begin{tabular}{lrr}
\hline
Variable &  $n$ & \% \\
\hline
\textbf{Respondent (child) characteristics:}         &           & \\
Age (years)                                         & Mean=43.7 & SD=11.4 \\[.5ex]
Gender                                              &           & \\
~~Female                                            & 7060      & 57.9 \\
~~Male                                              & 5143      & 42.1 \\[.5ex]
Partnership status                                  &           & \\
~~Partnered                                         & 9373      & 76.8 \\
~~Single                                            & 2830      & 23.2 \\[.5ex]
Age of youngest coresident child                        &           & \\
~~No children                  & 5002      & 41.0  \\
~~$0-1$ years                   & 910       & 7.5 \\
~~$2-4$ years                  & 1231      & 10.1 \\
~~$5-10$ years                 & 1910      & 15.7 \\
~~$11-16$ years                & 1548      & 12.7 \\
~~$17-$ years                  & 1602      & 13.1\\[.5ex]
Number of siblings             &           &\\
~~None                         & 1235      & 10.1 \\
~~1                            & 4325      & 35.4 \\
~~2 or more                           & 6643      & 54.4\\[.5ex]
Longstanding illness          &           &\\
~~Yes                          & 1533      & 12.6 \\
~~No                           & 10670     & 87.4\\[.5ex]
Employment status                                   &           & \\
~~Employed                                          &  9688     & 79.4 \\
~~Not employed                                       &  2515     & 20.6 \\
Education (highest qualification)                                           &           &\\
~~Secondary or less                          &  6024     & 49.4\\
~~Post-secondary                          &  6179     & 50.6\\[.5ex]
Household tenure                    &           &\\
~~Own home outright or with mortgage & 8817 & 72.3\\
~~Other (private or social renter)                    &  3386     & 27.7\\ [.5ex]
Logarithm of household equivalised income           & Mean=9.9  &
SD=0.79 \\ [8pt]
\textbf{Parent characteristics:}                     &           &
\\
Age of the oldest living parent (years)                    & Mean=72.1 & SD=11.2 \\[.5ex]
At least one parent lives alone                     &           & \\
~~Yes                                               & 4641      & 38.0 \\
~~No                                                & 7562      & 62.0 \\[8pt]
\textbf{Child--parent characteristics:} && \\
Travel time to the nearest parent                       &           & \\
~~1 hour or less                                    & 8851      & 72.5 \\
~~More than 1 hour                                  & 3352      & 27.5 \\
\hline
\multicolumn{3}{p{.8\textwidth}}{\footnotesize{Data from UKHLS, 2017-19
(Wave 9). The sample size for all covariates is $n=12,203$.}}
\end{tabular}
\end{table}

\clearpage

\section{Latent variable model for multivariate dyadic data} 
\label{sec:latent_variable_model}

Here we define the joint model that will be used to analyse the
multivariate dyadic data that were described in
\cref{sec:data_on_help_exchanges}. The model specification is broadly
similar to that of \cite{kuhaetal22}, but with two extensions. First,
tendencies to give and receive financial and practical help are
represented by separate latent variables, so that the model includes
four rather than two such variables for each respondent. Second, the
correlations between the latent variables are also modelled as functions
of covariates.

Let $(\mathbf{X}_i, \mathbf{Y}_{Gi}, \mathbf{Y}_{Ri})$ be observed data
for a sample of dyads $i=1,\ldots,n$,  where $\mathbf{X}_i$ is a
$Q\times 1$ vector of covariates (including a constant term 1), and
$\mathbf{Y}_{Gi}=(\mathbf{Y}_{GPi}^\top,Y_{GFi})^{\top}$ and
$\mathbf{Y}_{Ri}=(\mathbf{Y}_{RPi}^\top,Y_{RFi})^{\top}$ are
$(J+1)\times 1$ vectors of binary indicator variables (\textit{items}).
In our application, the dyads are those between a survey respondent and
his or her non-coresident parents, $\mathbf{Y}_{GPi}=(
Y_{GPi1},\dots,Y_{GPiJ})^{\top}$ are the respondent's answers to $J=7$
items on different types of \emph{practical help given} to their
parents, $\mathbf{Y}_{RPi}=( Y_{RPi1},\dots,Y_{RPiJ})^{\top}$ are the
items on \emph{practical help received} from the parents, and $Y_{GFi}$
and $Y_{RFi}$ are the single items on \emph{financial help given} and
\emph{financial help received} respectively. Each item is coded 1 if
that kind of help is given or received, and 0 if not. In other
applications, $\mathbf{Y}_{Gi}$ and $\mathbf{Y}_{Ri}$ could be of
different lengths and $Y_{GFi}$ and $Y_{RFi}$ could also be vectors of
multiple indicators, with straightforward modifications of the
specifications below.

\subsection{Measurement model for the observed items given latent variables}
\label{ss_models_measurement}

The items in $\mathbf{Y}_{GPi}$, $\mathbf{Y}_{RPi}$, $Y_{GFi}$ and
$Y_{RFi}$ are regarded as measures of continuous latent variables
$\eta_{GPi}$, $\eta_{RPi}$, $\eta_{GFi}$ and $\eta_{RFi}$ respectively.
Here we interpret $\eta_{GPi}$ and $\eta_{RPi}$ as an individual's
underlying tendencies to give and to receive practical help
respectively, and $\eta_{GFi}$ and $\eta_{RFi}$ similarly as tendencies
to give and receive financial help.

The data that we analyse have a large number of responses where all the
items in $\mathbf{Y}_{Gi}$ or $\mathbf{Y}_{Ri}$ are zero (no help given
or received; see \cref{t_items}). The proportions of these all-zero
responses may be higher than can be well accounted for by standard latent
variable models given the continuous latent variables alone. To allow
for this multivariate zero inflation, the model also includes two binary latent class
variables $\xi_{Gi}$ and $\xi_{Ri}$, for each of which one class represents
individuals who are certain not to give (for $\xi_{Gi}$) or receive (for
$\xi_{Ri}$) any kind of help. For giving help, the \textit{measurement
model} for the observed responses $\mathbf{Y}_{Gi}$
given the latent variables $(\eta_{GPi},\eta_{GFi},\xi_{Gi})$
is then specified by
\begin{eqnarray}
  p(\mathbf{Y}_{Gi}=\boldsymbol{0}| \xi_{Gi}=0,
  \eta_{GPi},\eta_{GFi};\boldsymbol{\phi}_G) &=&
  p(\mathbf{Y}_{Gi}=\boldsymbol{0}| \xi_{Gi}=0)
  =
  1\quad\text{and}\label{eq:meas1} \\
  p(\mathbf{Y}_{Gi}| \xi_{Gi}=1,
  \eta_{GPi},\eta_{GFi};\boldsymbol{\phi}_G) &=&
  \prod_{j=1}^J p(Y_{GPij}| \xi_{Gi}=1, \eta_{GPi};\boldsymbol{\phi}_G)
  \nonumber \\
   && \hspace*{3em}\times \;p(Y_{GFi}| \xi_{Gi}=1,\eta_{GFi}),
   \label{eq:meas2}
\end{eqnarray}
where $p(\cdot|\cdot)$ denotes a conditional distribution and
$\boldsymbol{\phi}_G$ are measurement parameters. When $\xi_{Gi}=0,$
respondent $i$ is thus certain to answer `No' to all items related to
giving help. When $\xi_{Gi}=1,$ the probabilities of responses to
$Y_{GPij}$ are determined by the continuous latent variable $\eta_{GPi}$
and the response to $Y_{GFi}$ is determined by $\eta_{GFi}$. Items $Y_{GPij}$ ($j=1,\dots,J$) are
assumed to be conditionally independent of each other given
$\eta_{GPi}$. If any items in $\mathbf{Y}_{Gi}$ are missing for
respondent $i$, they are omitted from the product in (\ref{eq:meas2}).
The measurement models for the individual items are specified as
\begin{eqnarray}
    p(Y_{GPij}=1| \xi_{Gi}=1, \eta_{GPi};\boldsymbol{\phi}_G) &=&
    \Phi(\tau_{GPj} + \lambda_{GPj}\, \eta_{GPi})
    \hspace*{.5em}\text{ for  }j=1,\dots,J,\hspace*{.5em}\text{and}
\label{eq:probit_linkP}
    \\
    p(Y_{GFi}=1| \xi_{Gi}=1,\eta_{GFi}) &=&
         \mathbbm{1}(\eta_{GFi}>0),
\label{eq:probit_linkF}
\end{eqnarray}
where $\Phi(\cdot)$ is the cumulative distribution function of the
standard normal distribution, $\mathbbm{1}(\cdot)$ is the indicator
function, $\tau_{GPj}$ and $\lambda_{GPj}$ are parameters, and we fix
$\tau_{GP1}=0$ and $\lambda_{GP1}=1$ for identification of the scale of
$\eta_{GPi}$. Here (\ref{eq:probit_linkP}) is a standard latent-variable
(item response theory) model for binary items, with probit measurement
models, and (\ref{eq:probit_linkF}), combined with the normal
distribution of $\eta_{GFi}$ defined below, is a latent-variable
formulation of a probit model for the single item $Y_{GFi}$. Thus
$\boldsymbol{\phi}_G =
(\tau_{GP2},\ldots,\tau_{GPJ},\lambda_{GP2},\ldots,\lambda_{GPJ})^{\top}$.
The measurement model for receiving help $\mathbf{Y}_{Ri}$ given
$(\eta_{RPi},\eta_{RFi},\xi_{Ri})$ is defined analogously to
(\ref{eq:probit_linkP})--(\ref{eq:probit_linkF}), with parameters
$\boldsymbol{\phi}_{R}$, and $\mathbf{Y}_{Gi}$ and $\mathbf{Y}_{Ri}$ are
assumed to be conditionally independent of each other given the latent
variables. Let
$\boldsymbol{\phi}=(\boldsymbol{\phi}_{G}^{\top},\boldsymbol{\phi}_{R}^{\top})^{\top}$.

\subsection{Structural model for the latent variables given covariates}
\label{ss_models_structural}

Let
$\boldsymbol{\eta}_{i}=(\eta_{GPi},\eta_{RPi},\eta_{GFi},\eta_{RFi})^{\top}$
and $\boldsymbol{\xi}_{i}=(\xi_{Gi},\xi_{Ri})^{\top}$.
Their conditional distribution
$p(\boldsymbol{\eta}_{i},\boldsymbol{\xi}_{i}|\mathbf{X}_{i};\boldsymbol{\psi})
=p(\boldsymbol{\eta}_{i}|\mathbf{X}_{i};\boldsymbol{\psi}_{\eta})\,
p(\boldsymbol{\xi}_{i}|\mathbf{X}_{i};\boldsymbol{\psi}_{\xi})$ is the
\emph{structural model} for the latent variables given the covariates. Here
$\boldsymbol{\eta}_{i}$ and $\boldsymbol{\xi}_{i}$ are taken to be
conditionally independent, and
$\boldsymbol{\psi}=(\boldsymbol{\psi}_{\eta}^{\top},\boldsymbol{\psi}_{\xi}^{\top})^{\top}$
are parameters. The distribution of the latent class
variables $\boldsymbol{\xi}_{i}$ is specified as multinomial, with
probabilities \begin{align}\label{eq:xi_dist}
\text{log}&\left[\frac{\pi_{k_1k_2}(\mathbf{X}_i)}{\pi_{00}(\mathbf{X}_i)}\right]
= \boldsymbol\gamma_{k_1k_2}^\top\mathbf{X}_i, \end{align} where
$\pi_{k_1k_2}(\mathbf{X}_i) = p(\xi_{Gi}=k_1,\xi_{Ri}=k_2|
\mathbf{X}_i; \boldsymbol{\psi}_{\xi})$ for $k_1,k_2=0,1$ and
$\boldsymbol{\gamma}_{00}=\boldsymbol{0}$, so that
$\boldsymbol{\psi}_{\xi}=( \boldsymbol{\gamma}_{01}^{\top},
\boldsymbol{\gamma}_{10}^{\top},
\boldsymbol{\gamma}_{11}^{\top})^{\top}$.

The structural model for the continuous helping tendencies
$\boldsymbol{\eta}_{i}$ given the covariates $\mathbf{X}_{i}$ is the
main focus of substantive interest in our analysis. Here $\boldsymbol{\eta}_{i}\sim N(\boldsymbol{\mu}_{i},
\boldsymbol{\Sigma}_{i})$ is taken to follow a four-variate normal distribution
with covariance matrix $\boldsymbol{\Sigma}_{i}$ and
mean vector
\begin{equation}
\boldsymbol{\mu}_{i}=\E(\boldsymbol{\eta}_{i}|\mathbf{X}_{i};\boldsymbol{\beta})
=
\begin{bmatrix}
\boldsymbol{\beta}_{GP}^{\top}\,\mathbf{X}_{i}\\
\boldsymbol{\beta}_{RP}^{\top}\,\mathbf{X}_{i}\\
\boldsymbol{\beta}_{GF}^{\top}\,\mathbf{X}_{i}\\
\boldsymbol{\beta}_{RF}^{\top}\,\mathbf{X}_{i}
\end{bmatrix}
= \boldsymbol{\beta}^{\top}\mathbf{X}_i
\label{eta_mean}
\end{equation}
where $\boldsymbol{\beta}=[\boldsymbol{\beta}_{GP},\,
 \boldsymbol{\beta}_{RP},\,\boldsymbol{\beta}_{GF},\,
 \boldsymbol{\beta}_{RF}]$ is a $Q\times 4$ matrix of coefficients.
In other words, (\ref{eta_mean}) specifies separate linear models for
the means of each element of $\boldsymbol{\eta}_{i}$. For the covariance
matrix, we first decompose it as
\begin{equation}
\boldsymbol{\Sigma}_{i}
=\text{cov}(\boldsymbol{\eta}_{i}|\mathbf{X}_{i};\boldsymbol{\alpha},\boldsymbol{\sigma})=
\mathbf{S}_{i}\, \mathbf{R}_{i}\, \mathbf{S}_{i},
\label{eta_var}
\end{equation}
where $\boldsymbol{\alpha}$ are parameters of the correlation matrix
$\mathbf{R}_{i}$ and
$\boldsymbol{\sigma}=(\sigma_{GP},\sigma_{RP})^{\top}$ are parameters of
$\mathbf{S}_{i}=\text{diag}(\sigma_{GP},\sigma_{RP}, 1, 1)$, a diagonal
matrix of standard deviations where those of $\eta_{GFi}$ and
$\eta_{RFi}$ are fixed at 1 to identify the measurement model
(\ref{eq:probit_linkF}) for $\eta_{GFi}$ and the corresponding model for
$\eta_{RFi}$. Here $\boldsymbol{\sigma}$ do not depend on the covariates
(and thus we could write $\mathbf{S}_{i}=\mathbf{S}$), but this
extension could also be included.

For the correlation matrix, we consider the specification
\begin{equation}
\mathbf{R}_{i}=
\mathbf{R}(\mathbf{X}_i;\boldsymbol\alpha) =
\begin{bmatrix}
    1 & & &\\
    \rho_{1i} & 1 & &\\
    \rho_{2i} &
    \rho_{4i} & 1 &\\
    \rho_{3i} &
    \rho_{5i} &
    \rho_{6i} & 1\\
  \end{bmatrix}
=\begin{bmatrix}
    1 & & &\\
    \rho(\mathbf{X}_i;\boldsymbol\alpha_1) & 1 & &\\
    \rho(\mathbf{X}_i;\boldsymbol\alpha_2) &
    \rho(\mathbf{X}_i;\boldsymbol\alpha_4) & 1 &\\
    \rho(\mathbf{X}_i;\boldsymbol\alpha_3) &
    \rho(\mathbf{X}_i;\boldsymbol\alpha_5) &
    \rho(\mathbf{X}_i;\boldsymbol\alpha_6) & 1\\
  \end{bmatrix},
\label{eta_corr}
\end{equation}
where only the lower triangular part is shown and the $L=6$ distinct
correlations are numbered as shown in (\ref{eta_corr}).
We specify separate linear models
$\rho_{li} = \rho(\mathbf{X}_i;\boldsymbol{\alpha}_l)=\boldsymbol{\alpha}_{l}^{\top}\mathbf{X}_{i}$ for each
$l=1,\dots,L$, i.e.\
\begin{equation}
\boldsymbol{\rho}_{i} = \boldsymbol{\alpha}^{\top}\mathbf{X}_{i}
\label{eta_corr_model}
\end{equation}
where $\boldsymbol{\rho}_{i}=(\rho_{1i},\dots,\rho_{Li})^{\top}$ and
$\boldsymbol{\alpha}=[\boldsymbol{\alpha}_{1}, \dots,
\boldsymbol{\alpha}_{L}]$ is a matrix of coefficients. Note that some
variables in $\mathbf{X}_{i}$ may be included in only one of the models
(\ref{eta_mean}) and (\ref{eta_corr_model}); if so, some of the
corresponding elements of $\boldsymbol{\beta}$ or $\boldsymbol{\alpha}$
are set to zero. The motivation and interpretation of this choice of
model for the correlation matrix is discussed further in Section
\ref{sec:model_corr}. The full set of parameters of the structural model
for $\boldsymbol{\eta}_{i}$ is thus
$\boldsymbol{\psi}_{\eta}=(\text{vec}(\boldsymbol{\beta})^{\top},
\boldsymbol{\sigma}^{\top},\text{vec}(\boldsymbol{\alpha})^{\top})^{\top}$,
where $\text{vec}(\cdot)$ denotes the vectorization of a matrix.

Let $\mathbf{Y}= [\mathbf{Y}_1,\ldots,\mathbf{Y}_n]^\top$
denote all the observed data on the items, where $\mathbf{Y}_i =
(\mathbf{Y}_{Gi}^\top,\mathbf{Y}_{Ri}^\top)^{\top}$, and
$\mathbf{X}=[\mathbf{X}_1,\ldots,\mathbf{X}_n]^\top$ the
data on the covariates. Define $G_i =
\mathbbm{1}(\mathbf{Y}_{Gi}\neq\boldsymbol{0})$ and $R_i =
\mathbbm{1}(\mathbf{Y}_{Ri}\neq\boldsymbol{0})$, the indicators for whether
responses on giving and on receiving help are not all zero for respondent $i$.
Assuming the observations for different respondents to be
independent, the log likelihood function of the model is
\begin{align*}
  &\log p(\mathbf{Y}|\mathbf{X}; \boldsymbol{\phi},\boldsymbol{\psi})\\
  &= \sum_{i=1}^N
  \log\Big\{\pi_{11}(\mathbf{X}_i;\boldsymbol{\psi}_\xi)
  \left[\int p(\mathbf{Y}_{Gi}|\xi_{Gi}=1,\eta_{GPi},\eta_{GFi};\boldsymbol{\phi}_G)
  p(\mathbf{Y}_{Ri}|\xi_{Ri}=1,\eta_{RPi},\eta_{RFi};\boldsymbol{\phi}_R)\right.\\
  &\qquad
  \hspace*{10em}\left.\phantom{\int}\times
  p(\boldsymbol{\eta}_{i}|\mathbf{X}_{i};\boldsymbol{\psi}_\eta)
  \;d\eta_{GPi}\,d\eta_{RPi}\,d\eta_{GFi}\,d\eta_{RFi}\right]\\
  &\quad + (1-R_i)\,\pi_{10}(\mathbf{X}_i;\boldsymbol\psi_\xi)\,\left[\int
  p(\mathbf{Y}_{Gi}|\xi_{Gi}=1,\eta_{GPi},\eta_{GFi};\boldsymbol{\phi}_G)\right.\\
  &\qquad
  \hspace*{10em}\left.\phantom{\int}\times
  p(\eta_{GPi},\eta_{GFi}|\mathbf{X}_{i};\boldsymbol{\psi}_\eta)\;d\eta_{GPi}\,d\eta_{GFi}\right]\\
  &\quad + (1-G_i)\,\pi_{01}(\mathbf{X}_i;\boldsymbol\psi_\xi)\,
  \left[\int
  p(\mathbf{Y}_{Ri}|\xi_{Ri}=1,\eta_{RPi},\eta_{RFi};\boldsymbol{\phi}_R)\right.\\
  &\qquad
  \hspace*{10em}\left.\phantom{\int}\times
  p(\eta_{RPi},\eta_{RFi}|\mathbf{X}_{i};\boldsymbol{\psi}_\eta)\;d\eta_{RPi}\,d\eta_{RFi}\right]\\
  &\quad + (1-G_i)(1-R_i)\,\pi_{00}(\mathbf{X}_i;\boldsymbol\psi_\xi)\Big\}.
\end{align*}
Estimation of this model is described in Section
\ref{sec:estimation_of_the_latent_variable_model}, after
some further discussion of questions related to the model for the correlations.

\section{Modelling correlation and covariance matrices given covariates: Existing
approaches}
\label{sec:model_corr}

There is a large literature on modelling association structures of
multivariate distributions. We review here those parts of it that are
most relevant to our work. This means that we consider different
approaches to modelling correlation or covariance matrices given
covariates, with a particular focus on how the models are specified.
This can be combined with different (parametric or other) specifications
for the joint distribution as a whole, and different methods of
estimating its parameters. Our own model uses a parametric specification
of a multivariate normal distribution and Bayesian estimation of its
parameters, but the review here is not limited to that case.

We consider only approaches which specify associations in terms of
covariances or correlations. This means that we exclude models for
conditional associations of some of the variables given the others, such
as log-linear models for categorical data or covariance selection models
for the inverse covariance matrix of a multivariate normal distribution.
We include here models for both correlation and covariance matrices,
although our model is for the correlation matrix. We focus on models
which are specified directly for these associations or transformations
of them. This excludes models where the associations are determined
indirectly via latent variables, such as random effects models and
common factor models. Note, however, that the multivariate response
variable whose covariance or correlation matrix is being modelled may
itself be a latent variable, as it is in our analysis where we model the
correlations of the latent $\boldsymbol{\eta}_{i}$.

Models for associations may have two broad goals. The first of them is
to impose a patterned structure on the covariance or correlation matrix
which is more parsimonious than an unstructured matrix that has separate
parameters for each pair of variables. This is the case, for example,
when an autocorrelation structure is specified for responses that are
ordered in time. An extreme version of this occurs in very
high-dimensional problems where parsimonious specification is essential
for consistent estimation of covariance matrices. We do not consider
such regularisation methods here (see \citealt{pourahmadi2011} and
\citealt{fanetal16} for reviews). The second broad type of model
specification considers instead an unstructured model of associations,
but specifies how the correlations or covariances in it depend on
covariates that are characteristics of the units of analysis, such as
the survey respondents in our application. This is the goal of our
modelling.

A key question is how to ensure that the estimated matrices will be
positive definite. Here \citet{pinheiro1996unconstrained} pointed out a
key distinction between two approaches: unconstrained ones where the
models are specified for parametrizations (transformations) of the
association matrix which are guaranteed to imply a positive definite
matrix, and constrained ones where positive definiteness is imposed in
the estimation process. Our approach is an instance of constrained
estimation, but we summarise first the most important unconstrained
methods (see \citealt{pourahmadi2011} and \citealt{pan+pan17} for more
detailed reviews). They differ in what transformation they use. The most
common is the modified Cholesky decomposition of the covariance matrix.
It was introduced by \cite{pourahmadi1999joint}, and general models for
it (including covariates) were proposed by \cite{pan+mackenzie06}. Other
possible transformations include the matrix logarithm of the covariance
matrix \citep{chiu1996matrix}, the `alternative Cholesky decomposition'
of the covariance matrix \citep{chen2003random}, a variant of the
modified Cholesky decomposition proposed by \cite{zhang+leng12},
parametrizations of the correlation matrix in terms of partial
autocorrelations \citep{wang+daniels13} or hyperspherical co-ordinates
of its standard Cholesky decomposition \citep{zhangetal15}, and the
matrix logarithm of the correlation matrix
\citep{archakov2021new,hu2021regression}.

The natural advantage of the unconstrained methods is that they ensure
positive definiteness at any values of the covariates. The corresponding
disadvantage is that because the models are not specified for the
individual association parameters (or even transformations of them), the
model parameters are not easily interpretable. All of the
interpretations that are available apply only in cases where the
response variables have a natural ordering, most obviously in
longitudinal data where they are ordered in  time. Then the parameters
of the modified Cholesky decomposition can be interpreted in terms of an
autoregressive model for each variable given its predecessors, and those
of the alternative Cholesky decomposition and of \cite{zhang+leng12}
similarly in terms of a moving average representation of each variable
given error terms of the previous ones
\citep{pourahmadi2007cholesky,pan+pan17}, those of \cite{wang+daniels13}
as partial autocorrelations of two variables given all the intervening
ones, and the hyperspherical co-ordinate parametrization in terms of
semi-partial correlations \citep{ghoshetal21}. In our application these
properties are not helpful because we consider response variables with
no ordering and want to obtain a simple interpretation of coefficients
for the correlations themselves.

Turning now to approaches that model individual pairwise association
parameters directly, for correlations we could use transformations of
them (e.g.\ Fisher's $z$) to ensure that the fitted correlations are
constrained to $(-1,1)$. This, however, is not sufficient to ensure that
the correlation matrix as a whole is positive definite, except for a
bivariate response (for this case, see e.g. \citealt{wildingetal11} and
references therein). One pragmatic approach that we could then take is
to simply employ such models anyway also more generally, ignoring the
possibility of some non-positive definite matrices (see e.g.
\citealt{yan+fine07}). It is plausible that this will work well in some
applications, in the best case that the fitted correlation matrices end
up being positive definite at all relevant values of the covariates.
However, it is not in principle a satisfactory general approach.
\cite{luo+pan22} suggest post-hoc adjustments to fitted correlation
models to make them positive definite; this, however, is unhelpful when
the focus is on interpreting coefficients of the model. A different
solution is provided by \cite{hoff+niu12} who propose a model (analogous
in form to factor analysis models) where covariances depend on quadratic
functions of covariates, and the matrix is automatically positive
definite.

Existing literature on constrained estimation focuses on linear models
for covariances or correlations. This is not really a limitation even
for correlations, because the constraint that the matrix should be
positive definite also implies that the individual correlations will be
in $(-1,1)$. The most developed results here are for the linear
covariance model for multivariate normal distribution
\citep{anderson1973asymptotically}, in which  the covariance matrix
takes the form $\boldsymbol{\Sigma}=\sum_{k} \nu_{k} \mathbf{G}_{k}$
where $\nu_{k}$ are parameters and $\mathbf{G}_{k}$ are known, linearly
independent symmetric matrices. \cite{zwierniketal17} show that although
the log-likelihood function for this model typically has multiple local
maxima, any hill climbing method initiated at the least squares
estimator will converge to its global maximum with high probability.
\cite{zouetal17} consider the case where the $\mathbf{G}_{k}$ are
similarity matrices between the response variables, and propose maximum
likelihood estimates and constrained least squares estimators for this
model. In these formulations, $\boldsymbol{\Sigma}$ is the same for all
units $i$. This is further relaxed by \cite{zouetal22}, who allow the
values of the similarity matrices in \cite{zouetal17} to depend on
unit-specific covariates, thus defining $\boldsymbol{\Sigma}_{i}$ as
linear combinations of unit-specific similarity matrices.

We will also consider a linear model, as shown in
(\ref{eta_corr_model}), but for the correlations and given unit-specific
covariates. We will estimate it in the Bayesian framework and using
Markov chain Monte Carlo (MCMC) estimation. As MCMC updates different
parameters separately, it is natural to employ the decomposition
(\ref{eta_var}) and model the standard deviations and correlations
separately, and this is what has been done in most previous literature
that has used a Bayesian approach. MCMC also provides an obvious way to
implement constrained estimation, at least in principle. This can be
done at each sampling step of the estimation algorithm, by constraining
the prior distribution, the proposal distribution from which values of
the parameters are drawn, or the acceptance probabilities of the sampled
values, in a way which rules out inadmissible values of the sampled
parameters. But although the principle is obvious, implementing it is
not necessarily easy. Two instances of this approach that we will draw
on in particular are those of \cite{johnbar2000} and
\cite{wongEfficientEstimationCovariance2003}, and we will return to them
below. Other methods of this kind have been proposed by
\cite{chibAnalysisMultivariateProbit1998} and \cite{liechtyetal04}.

What is missing from these existing Bayesian implementations is the
inclusion of unit-specific covariates in the models for the
correlations, which is our focus. In Section
\ref{sec:estimation_of_the_latent_variable_model} we propose an extended
estimation procedure which accommodates such covariates. This in turn
requires some further consideration of the constraints on positive
definiteness of the correlation matrix, because this now has to hold at
different values of the covariates. This question is discussed first, in
the next section.

\section{Ensuring a positive definite correlation matrix}  
\label{sub:Some Useful Properties}

The key challenge in our constrained estimation is to ensure that the
estimated correlation matrices remain positive definite at all relevant
values of the covariates. Here we give some further theoretical results
which establish when and how this can be achieved.

Let $\mathbf{R}=\mathbf{R}(\boldsymbol{\rho})$ denote a symmetric matrix
where all the diagonal elements equal 1 and the distinct off-diagonal
elements $\boldsymbol{\rho}=(\rho_{1},\dots,\rho_{L})^{\top}$ are all in
$(-1,1)$. Let $C_{\rho}$ denote the set of $\boldsymbol{\rho}$ such that
$\mathbf{R}(\boldsymbol{\rho})$ is positive definite, and thus a
correlation matrix, for all $\boldsymbol{\rho}\in C_{\rho}$. It is a
convex subset of the hypercube $[-1,1]^L$ \citep[for the shape of
$C_{\rho}$ in the cases $L=3$ and $L=6$, i.e.\ for $3\times3$ and
$4\times 4$ correlation matrices,
see][]{rousseeuwShapeCorrelationMatrices1994}.

We consider model (\ref{eta_corr_model}) where
$\boldsymbol{\rho}=\boldsymbol{\alpha}^{\top}\mathbf{X}$. In this
section we mostly omit the unit subscript from $\mathbf{X}$ and
$\boldsymbol{\rho}$, and take $\mathbf{X}$ to include only those
covariates that are included in the model for $\boldsymbol{\rho}$
(excluding any that are used only for the means $\boldsymbol{\mu}$). It
is clear that this $\boldsymbol{\rho}$ cannot be in $C_{\rho}$ for all
values of the parameters $\boldsymbol{\alpha}$ and covariates
$\mathbf{X}$. We need to limit the scope of the estimated models to the
values that do imply $\boldsymbol{\rho}\in C_{\rho}$. For the
covariates, it is useful to introduce here some additional notation. Let
$\mathbf{Z}$ be the smallest vector of distinct variables, including a
constant term 1, which determines $\mathbf{X}=\mathbf{X}(\mathbf{Z})$.
Here $\mathbf{Z}$ may be shorter than $\mathbf{X}$ if some variables in
$\mathbf{X}$ are functions of a smaller number of variables in
$\mathbf{Z}$, e.g.\ when $\mathbf{X}$ includes polynomials or product
terms (interactions). Suppose that $\mathbf{Z}$ is a $p \times 1$ vector
and $\mathbf{X}$ a $q\times 1$ vector. Below we denote sets
$S_{Z}\subset \mathbb{R}^{p}$ and $S_{X}\subset \mathbb{R}^{q}$ of
$\mathbf{Z}$ and $\mathbf{X}$ respectively with appropriate subscripts,
and also $S_{X}=\mathbf{X}(S_{Z})=\{\mathbf{X}(\mathbf{Z})\mid
\mathbf{Z}\in S_Z\}$ when needed to indicate how a set of $\mathbf{Z}$
determines that of $\mathbf{X}$.

A combination of values $(\mathbf{Z},\boldsymbol{\alpha})$ is said
to be \emph{feasible} if
$\boldsymbol{\rho}=\boldsymbol{\alpha}^{\top}\mathbf{X}(\mathbf{Z})\in
C_{\rho}$, and $(\mathbf{X},\boldsymbol{\alpha})$ to be feasible if
$\boldsymbol{\rho}=\boldsymbol{\alpha}^{\top}\mathbf{X}\in C_{\rho}$. We
aim to identify known sets of $\mathbf{Z}$ and $\boldsymbol{\alpha}$
such that all combinations of values from them are feasible. This will involve the following steps:
\begin{itemize}
\item[(1)] Specify a set $S_{Z}$
for $\mathbf{Z}$ for which we want the
ensure that the estimated correlation matrices are positive definite.
\item[(2)]
Specify a finite test
set $S_{XT}=\{\mathbf{X}_{1},\dots,\mathbf{X}_{T}\}$ of
values for $\mathbf{X}$, which will be used for checking positive
definiteness during MCMC estimation. The choice of $S_{XT}$ will depend on $S_Z$.
\item[(3)]
Carry out MCMC estimation which includes sampling values of
$\boldsymbol{\alpha}$ (together with the other model parameters). At
each iteration, carry out checks to ensure that the value of
$\boldsymbol{\alpha}$ that is retained from the iteration is feasible with all
$\mathbf{X}\in S_{XT}$. In the end, this produces an MCMC sample
$S_{\alpha}=\{\boldsymbol{\alpha}_{1},\dots,\boldsymbol{\alpha}_{M}\}$.
\item[(4)] Conclude that $(\mathbf{Z},\boldsymbol{\alpha})$ is feasible
for all combinations of $\mathbf{Z}\in S_{Z}$ and values of
$\boldsymbol{\alpha}$ in the convex hull of $S_{\alpha}$.
\end{itemize}
In step (1), $S_{Z}$ should include the substantively relevant and
interesting values of the covariates for which we want our estimated
model to imply valid correlation matrices. For example, this could be a
finite set $S_{ZN}=\{\mathbf{Z}_{1},\dots,\mathbf{Z}_{N}\}$, normally
including at least all the distinct values among the $\mathbf{Z}_{i}$,
$i=1,\dots,n$, in the observed data. $S_{Z}$ is always of this form when
all the variables in $\mathbf{Z}$ are categorical (coded as dummy
variables). If $\mathbf{Z}$ includes continuous variables, we can also
expand $S_{ZN}$ to an infinite set, such as its convex hull $
S_{Zh}=\{\sum_{j=1}^{N} \lambda_{j}
\mathbf{Z}_{j}\;|\;\sum_{j=1}^{N}\lambda_{j}=1; \lambda_{j}\ge 0 \text{
for all }j=1,\dots,N \}$ or the hyperrectangle $S_{Zr} =
\left\{(Z_{1},\dots,Z_{p})\; | \; Z_{s}\in[l_{s},u_{s}] \text{ for all
}s=1,\dots,p\right\}, $ for specified $l_{s}\le \min
\{Z_{js}\;|\;j=1,\dots,N\}$ and $u_{s}\ge \max
\{Z_{js}\;|\;j=1,\dots,N\}$ for each $s=1,\dots, p$. Note that here
$S_{ZN}\subset S_{Zh}\subseteq S_{Zr}$.

How the sample $S_{\alpha}$ in step (3) can be drawn is described at the
end of this section and in Section
\ref{sec:estimation_of_the_latent_variable_model}. That alone is not yet
enough for what we need. This is because in step (3) we can only check
feasibility for a finite number of values of $\mathbf{X}$ and
$\boldsymbol{\alpha}$, whereas the sets we want to draw conclusions on
are infinite for at least $\boldsymbol{\alpha}$ and possibly also for
$\mathbf{Z}$. So some additional theoretical results are needed to
justify the conclusion in step (4); this also informs the choice of the
test set $S_{XT}$ in step (2).

Consider first the correlations
$\boldsymbol{\rho}= \boldsymbol{\alpha}^{\top}\mathbf{X}$ as they depend
on $\mathbf{X}$ rather than $\mathbf{Z}$. Let
\[
C_{\alpha,S_X} =
\{\boldsymbol\alpha\in\mathbb{R}^{L\times q} \, | \, \boldsymbol\rho =
\boldsymbol{\alpha}^{\top}\mathbf{X} \in C_\rho \; \text{ for all }\;
\mathbf{X}\in S_X\}
\]
be the set of values of $\boldsymbol{\alpha}$ which are feasible when combined with
any $\mathbf{X}$ in $S_{X}$. \cref{prop:C_mu1} gives some basic
properties of $C_{\alpha,S_{X}}$. Proofs of them are given in
\ref{app:proofs_of_lemmas}.

\begin{proposition}
Properties of $C_{\alpha,S_X}$:
\label{prop:C_mu1}\leavevmode
  \begin{enumerate}[(i)]
    \item If $S_{X_2}\subseteq S_{X_1},$ then
    $C_{\alpha,S_{X_1}}\subseteq C_{\alpha,S_{X_2}}$.
    \item
    $C_{\alpha,\text{Conv}(S_X)}=C_{\alpha,S_X}$, where $\text{Conv}(S_X)$
    denotes the convex hull of $S_X$.
    \item $\boldsymbol 0 \in C_{\alpha,S_X}$.
    \item
    Suppose further that  there exist $q$ linearly independent elements in $S_X$. Then
        $C_{\alpha,S_X}$ is bounded.
    \item $C_{\alpha,S_X}$ is a convex set.
  \end{enumerate}
\end{proposition}

Parts (i) and (ii) of Proposition \ref{prop:C_mu1} explain how
$C_{\alpha,S_{X}}$ depends on the set $S_{X}$ of values considered for
$\mathbf{X}$. When step (3) is completed, we know that
$\boldsymbol{\alpha}\in C_{\alpha,S_{XT}}$ for all
$\boldsymbol{\alpha}\in S_{\alpha}$. Then also $\boldsymbol{\alpha}\in
C_{\alpha,\text{Conv}(S_{XT})}$ by (ii). In other words, even though
feasibility was checked only for a finite number of values of
$\mathbf{X}$, we know that it holds also for the infinite set of their
convex hull.

We then need to translate this result for $\mathbf{X}$ back to
$\mathbf{Z}$. This is simple if $\mathbf{X}=\mathbf{Z}$, so that
$S_{XT}=S_{ZT}$, where $S_{ZT}$ is a finite test set for $\mathbf{Z}$.
Here we need to ensure that $S_{ZT}$ has been chosen so that
$S_{Z}\subseteq \text{Conv}(S_{ZT})$, i.e.\ that the convex hull of
$S_{ZT}$ covers $S_{Z}$. Then, for any $\boldsymbol{\alpha}\in
S_\alpha$, we have $\boldsymbol{\alpha}\in
C_{\alpha,\text{Conv}(S_{ZT})}$ by (ii) as above, and then
$\boldsymbol{\alpha}\in C_{\alpha,S_{Z}}$ by (i), as required. In terms
of the possible target sets $S_{Z}$ defined above, the test set $S_{ZT}$
could be $S_{ZN}$, which ensures feasibility also for all $\mathbf{Z}$
in $S_{Zh}$, or $S_{ZT}$ could consist of the vertices of $S_{Zr}$,
which ensures feasibility in all of $S_{Zr}$, $S_{Zh}$ and $S_{ZN}$.

Some more care is needed when $\mathbf{X}=\mathbf{X}(\mathbf{Z})$
includes non-linear functions of $\mathbf{Z}$. If $S_{Z}=S_{ZN}$ is
finite, a simple pragmatic choice is to set $S_{XT}= \mathbf{X}(S_{Z})$
and check all of their values. Otherwise, the forms of these functions
need to be considered. This can be seen already in a model for a single
correlation $\rho$, e.g.\ when it is modelled as a quadratic function
$\rho=\alpha_{0}+\alpha_{1}Z+\alpha_{2}Z^{2}$ of a single $Z$, so that
$\mathbf{X}(Z)=(X_{1},X_{2},X_{3})^\top=(1,Z,Z^{2})^\top$. Suppose that
$S_{Z}=[Z_{1},Z_{2}]$. It is not enough to check just the points
$\mathbf{X}_{1}=\mathbf{X}(Z_1)$ and $\mathbf{X}_{2}=\mathbf{X}(Z_2)$,
because we may have $\rho\in (-1,1)$ at $Z_{1}$ and $Z_{2}$ but not at
all values between them. It is sufficient to identify one more point
$\mathbf{X}_{3}=(1,X_{23},X_{33})^\top$ such that the convex hull of
$S_{XT}=\{\mathbf{X}_{1},\mathbf{X}_{2},\mathbf{X}_{3}\}$ covers
$\mathbf{X}(S_{Z})$. For example, this can be obtained by adding the
intersection point of the tangents of the curve $f(Z)=Z^{2}$ drawn at
$Z_{1}$ and $Z_{2}$, i.e.\ $X_{23}=(Z_{1}+Z_{2})/2$ and
$X_{33}=Z_{1}Z_{2}$. Note that such a choice depends only on the forms
of $S_{Z}$ and $\mathbf{X}(\mathbf{Z})$, so it can be used with any
value of $\boldsymbol{\alpha}$ and for any number of
correlations~$\rho_{l}$.

At this point we know that $\boldsymbol{\alpha}\in
C_{\alpha,\mathbf{X}(S_{Z})}$ for all $\boldsymbol{\alpha}\in
S_{\alpha}$, i.e.\ that all the values in the MCMC sample
$S_{\alpha}=\{\boldsymbol{\alpha}_{1},\dots,\boldsymbol{\alpha}_{M}\}$
are feasible when combined with any value of $\mathbf{Z}$ in the
(possibly infinite) target set $S_{Z}$. But we still need to extend this
conclusion to other values of $\boldsymbol{\alpha}$ that were not
sampled, specifically to the convex hull $\text{Conv}(S_{\alpha})$ of
$S_{\alpha}$. This is justified by parts (iii)--(v) of Proposition
\ref{prop:C_mu1}, which concern the values of $\boldsymbol{\alpha}$ in
$C_{\alpha,S_{X}}$ given a fixed $S_{X}$. Part (iii) shows that this set
is non-empty, so some feasible $\boldsymbol{\alpha}$ always exist, and
(iv) states that feasible $\boldsymbol\alpha$ will not drift away, as
long as $S_X$ is not degenerate. Finally, part (v) shows that
$\boldsymbol{\alpha}\in C_{\alpha,\mathbf{X}(S_{Z})}$ for all
$\boldsymbol{\alpha}\in \text{Conv}(S_{\alpha})$ as required, thus
completing step (4) of the process defined above. In particular, the
convex hull of the MCMC sample $S_{\alpha}$ includes the summary
statistics that we will typically use to summarise it for estimation of
the parameters in $\boldsymbol{\alpha}$, such as their (posterior) means
and quantile-based interval estimates. These estimates are thus also
guaranteed to imply positive definite correlation matrices given any
values of the covariates in the pre-specified set of interest $S_{Z}$.

We note that part (v) of Proposition 1 would not necessarily hold if the
individual correlations $\rho_{l}$ were modelled using a nonlinear
transformation, for example Fisher's~$z$ transformation which would give
the model
$\boldsymbol{\rho}=\text{tanh}(\boldsymbol{\alpha}^{\top}\mathbf{X})$.
Inferring feasibility from $S_{\alpha}$ to $\text{Conv}(S_{\alpha})$ is
thus not necessarily justified for these models. Such transformations
are normally used to ensure that individual correlations are constrained
to be in the range $(-1,1)$. That, however, is not needed here, because
positive definiteness of the matrix as a whole already implies that all
of the correlations are in the valid range.

Consider finally step (3), where we want to ensure that the value of
parameters $\boldsymbol{\alpha}$ retained from each MCMC iteration is
feasible with all $\mathbf{X}\in S_{XT}$. The key result for it is given
here, and it is then implemented as part of our estimation as described
in Section~\ref{sec:estimation_of_the_latent_variable_model}. Suppose
that $\boldsymbol{\alpha}'$ is a proposed value for
$\boldsymbol{\alpha}$. Since $S_{XT}$ is finite, it would be possible to
simply calculate $\mathbf{R}((\boldsymbol{\alpha}')^{\top}\mathbf{X})$
for all $\mathbf{X}\in S_{XT}$ and retain $\boldsymbol{\alpha}'$ if
these were all positive definite. This, however, could be
computationally demanding and inefficient. Instead, we will check and
update $\boldsymbol{\alpha}$ one element at a time. This builds on the
result by \cite{johnbar2000} that when starting with a correlation
matrix, there exists a continuous interval for each single correlation
in it that still yields a positive definite correlation matrix while
holding the rest of the correlations fixed at their previous values.
Here we extend this idea to apply to the individual coefficients in
$\boldsymbol{\alpha}$ and multiple values of the covariates
$\mathbf{X}$. The procedure relies on the following result, the proof of
which is given in \ref{app:proofs_of_lemmas}:

\begin{proposition} Continuous feasible intervals for the coefficients
in $\boldsymbol{\alpha}$:
\leavevmode\label{prop:conti_interval_mu}
Consider a finite set $S_{XT}=\{\mathbf{X}_j=(X_{j1},\dots,X_{jq})^{\top}\, |\,  j=1,\dots,T\}$
and any fixed value $\boldsymbol\alpha=[\boldsymbol{\alpha}_{1}, \dots,
\boldsymbol{\alpha}_{L}]^\top \in C_{\alpha,S_{XT}}$.
Denote here
(deviating slightly from previous notation)
$\boldsymbol{\alpha}=(\alpha_{lm},\boldsymbol\alpha_{-lm}^\top)^\top$
where $\alpha_{lm}$ is
the coefficient of $X_{jm}$ in the model for correlation $\rho_{l}$, for
any
$m=1,\dots,q$ and $l=1,\dots,L$, and $\boldsymbol{\alpha}_{-lm}$
denotes all other elements of $\boldsymbol{\alpha}$,
$\boldsymbol{\rho}=(\rho_{l},\boldsymbol{\rho}_{-l}^{\top})^{\top}$
where $\boldsymbol{\rho}_{-l}$ denotes all other elements of the
distinct correlations $\boldsymbol{\rho}$ except $\rho_{l}$, and
$\mathbf{R}(\rho_{l},\boldsymbol{\rho}_{-l})$ the correlation matrix
implied by $\boldsymbol{\rho}$. Let $\boldsymbol{\rho}_{-l}^{(j)}$
denote the value of $\boldsymbol{\rho}_{-l}$ for
$\boldsymbol{\rho}_{j}=\boldsymbol{\alpha}^{\top}\mathbf{X}_{j}$,
for $j=1,\dots,T$.
  \begin{enumerate}[(i)]
    \item
    There exists a non-empty interval $(a_{lm},b_{lm})$ such that
    $\boldsymbol\alpha' =
    (\alpha_{lm}',\boldsymbol\alpha_{-lm}^\top)^\top \in C_{\alpha,S_{XT}}$ for all
    $\alpha_{lm}'\in (a_{lm}, b_{lm})$.

    \item
Let $f_{jl}(\rho_{l}')=\vert\mathbf{R}(\rho_{l}',\boldsymbol\rho_{-l}^{(j)})\vert$,
treated as a function of $\rho_{l}'$, where $\vert\cdot\vert$
denotes the determinant of a matrix.
    If $X_{jm}\neq 0$, let
    \begin{equation}\label{eq:interval_j}
      \begin{aligned}
        a_{lm}^{(j)} &= \frac{g_{jl} - \sum_{k\neq
        m}\alpha_{lk}\,X_{jk} - \text{sgn}(X_{jm})\,h_{jl}}{X_{jm}},\\
        b_{lm}^{(j)} &= \frac{g_{jl} - \sum_{k\neq
        m}\alpha_{lk}\,X_{jk} + \text{sgn}(X_{jm})\,h_{jl}}{X_{jm}}
      \end{aligned}
    \end{equation}
    for each $j=1,\dots,T$,
where $g_{jl}=-d_{jl}/(2c_{jl})$ and
$h_{jl}=[(d_{jl}^{2}-4c_{jl}e_{jl})/(4c_{jl}^{2})]^{1/2}$
for $c_{jl} = [f_{jl}(1)+f_{jl}(-1)-2f_{jl}(0)]/2$, $d_{jl} = [f_{jl}(1) -
f_{jl}(-1)]/2$ and $e_{jl} = f_{jl}(0)$.
    If $X_{jm}=0$,
    set $a_{lm}^{(j)}=-\infty$ and $b_{lm}^{(j)}=+\infty$. Then
    $(a_{lm},b_{lm})=\cap_{j=1}^T (a_{lm}^{(j)},b_{lm}^{(j)})$.
This interval is non-empty because it contains at least
the current value $\alpha_{lm}$.
  \end{enumerate}
\end{proposition}

Computationally the most demanding part of using this result is the
calculation of the necessary determinants. Efficient methods for
obtaining them, and other elements of the computations, are discussed in
Section \ref{sub:estimation_of_the_structural_model2}.

\section{Estimation of the model} 
\label{sec:estimation_of_the_latent_variable_model}

Following the example and motivation of \cite{kuhaetal22}, we use a two-step procedure to estimate the latent variable model defined in Section
\ref{sec:latent_variable_model}.
The parameters of the measurement model are estimated first, as explained in
\cref{sub:estimation_of_the_measurement_model}. They are then fixed at
their estimated values in the second step, where the parameters of the
structural model are estimated as described in
\cref{sub:estimation_of_the_structural_model2}.

\subsection{Estimation of the measurement model} 
\label{sub:estimation_of_the_measurement_model}


In the first step, the
measurement parameters $\boldsymbol{\phi}_G$ and
$\boldsymbol{\phi}_R$ are estimated separately. For
$\boldsymbol{\phi}_G$, the data are the responses $\mathbf{Y}_{Gi}$,
the measurement model is specified by
(\ref{eq:meas1})--(\ref{eq:probit_linkF}), and the structural model for
$\xi_{Gi}$ and $\boldsymbol{\eta}_{Gi}=(\eta_{GPi},\eta_{GFi})^\top$ is obtained from
(\ref{eq:xi_dist})--(\ref{eta_corr}) by omitting
$\boldsymbol{\eta}_{Ri}$ and the covariates $\mathbf{X}_{i}$.
The log likelihood function for $\boldsymbol{\phi}_{G}$ is then
\begin{align}
  &\log p(\mathbf{Y}_{G}|\boldsymbol{\phi}_{G}
  ,\pi_G,
  \mu_{\eta_{GP}},
  \mu_{\eta_{GF}},
  \sigma^2_{\eta_{GP}}, \rho_{\eta_G}) \nonumber \\
  &= \sum_{i=1}^n
  \log\Big[ \pi_{G}\,
  \int
  p(\mathbf{Y}_{Gi}|\xi_{Gi}=1,\eta_{GPi},\eta_{GFi};\boldsymbol{\phi}_G)\,
  p(\boldsymbol{\eta}_{Gi}; \mu_{\eta_{GP}},
  \mu_{\eta_{GF}},
  \sigma^{2}_{\eta_{GP}}, \rho_{\eta_{G}})\,d\eta_{GPi}\, d\eta_{GFi}
  \nonumber \\
  &\quad \hspace*{2em}+ (1-G_i)(1-\pi_{G})\Big]
\label{step1_ll}
  \end{align}
where $\pi_G = p(\xi_{Gi}=1)$ and $p(\boldsymbol{\eta}_{Gi};
\mu_{\eta_{GP}}, \mu_{\eta_{GF}},
\sigma^{2}_{\eta_{GP}}, \rho_{\eta_{G}})$ is a bivariate
normal density with $\E(\eta_{GPi})=\mu_{\eta_{GP}}$,
$\Var(\eta_{GPi})=\sigma^{2}_{\eta_{GP}}$,
$\E(\eta_{GFi})=\mu_{\eta_{GF}}$, $\Var(\eta_{GFi})=1$ and
$\Corr(\eta_{GPi},\eta_{GFi})=\rho_{\eta_{G}}$.
The estimates
$\tilde{\boldsymbol{\phi}}_{G}$ of $\boldsymbol{\phi}_{G}$ are obtained
by maximizing (\ref{step1_ll}),
while the estimates of $\mu_{\eta_{GP}}$,
$\mu_{\eta_{GF}}$,
$\sigma^{2}_{\eta_{GP}}$ and
$\rho_{\eta_{G}}$ from this step are discarded. The estimates
$\tilde{\boldsymbol{\phi}}_{R}$ of $\boldsymbol{\phi}_{R}$ are obtained
analogously, using the data on $\mathbf{Y}_{Ri}$. We have used the
Mplus 6.12 software \citep{muthen2017mplus} to carry out this first step
of estimation.

\subsection{Estimation of the structural model} 
\label{sub:estimation_of_the_structural_model2}

In the second step of estimation, the structural parameters
$\boldsymbol{\psi}$ are estimated, treating the measurement parameters
fixed at their estimated values
$\tilde{\boldsymbol{\phi}}=(\tilde{\boldsymbol{\phi}}_{G}^\top,
\tilde{\boldsymbol{\phi}}_{R}^\top)^\top$. Here we omit
$\tilde{\boldsymbol{\phi}}$ from the notation for simplicity.

We use a Bayesian approach, implemented using MCMC methods. The estimation algorithm has a data augmentation structure which
alternates between imputing the latent variables given the observed
variables and values of the parameters, and sampling the parameters from
their posterior distributions given the observed and latent variables:

\begin{itemize}
\item
\textbf{Sampling of the latent variables}:
Let $\boldsymbol{\zeta}=(\boldsymbol{\xi},\boldsymbol{\eta})$, where
$\boldsymbol{\xi}$ denotes all the values of the latent
$\boldsymbol{\xi}_{i}$ for the units $i$ in the sample, and
$\boldsymbol{\eta}$ all the values of $\boldsymbol{\eta}_{i}$. Given the
observed data $(\mathbf{Y},\mathbf{X})$ and the most recently sampled
values of the parameters $\boldsymbol{\psi}$, sample
$\boldsymbol{\zeta}$ from
\[
  p(\boldsymbol\zeta | \mathbf{Y},\mathbf{X},  \boldsymbol\psi) \,
  \propto \, p(\mathbf{Y}|\boldsymbol\zeta)\, p(\boldsymbol\zeta | \mathbf{X}, \boldsymbol
  \psi).
\]
\item
\textbf{Sampling of the parameters}:
Given the observed data on
$\mathbf{X}$ and the most recently sampled values of the
latent variables
$\boldsymbol{\zeta}$, sample $\boldsymbol{\psi}$ from the posterior
distribution
\[
  p(\boldsymbol\psi | \boldsymbol \zeta,\mathbf{X}) \,\propto\,
  p(\boldsymbol\zeta|\mathbf{X},\boldsymbol\psi)\, p(\boldsymbol\psi),
\]
where
\[
p(\boldsymbol\zeta | \mathbf{X},\boldsymbol\psi) =
p(\boldsymbol\eta|\mathbf{X}; \boldsymbol\psi_\eta)\, p(\boldsymbol\xi |
\mathbf{X}; \boldsymbol\psi_\xi)
\]
is specified by the structural model and
$p(\boldsymbol\psi)$ is the prior distribution of
$\boldsymbol{\psi}=(\boldsymbol{\psi}_{\eta}^\top,\boldsymbol{\psi}_{\xi}^\top)^\top$.
We take $\boldsymbol{\psi}_{\eta}$ and $\boldsymbol{\psi}_{\xi}$ to be
independent a priori, so that
$p(\boldsymbol{\psi})=p(\boldsymbol\psi_\eta)p(\boldsymbol \psi_\xi)$.
The posterior distribution then divides into separate posteriors
for $\boldsymbol{\psi}_{\eta}$ and $\boldsymbol{\psi}_{\xi}$ as
\[
  p(\boldsymbol\psi | \boldsymbol \zeta,\mathbf{X})=
  p(\boldsymbol\psi_{\eta} | \boldsymbol \eta,\mathbf{X})\,
  p(\boldsymbol\psi_{\xi} | \boldsymbol \xi,\mathbf{X})
  \,\propto\,
  [p(\boldsymbol\eta | \mathbf{X}; \boldsymbol{\psi}_{\eta})
  p(\boldsymbol\psi_\eta)]\,
  [p(\boldsymbol\psi | \mathbf{X}; \boldsymbol{\psi}_{\xi})
  p(\boldsymbol\psi_\xi)].
\]
\end{itemize}
These steps further split into separate steps for different
components of $\boldsymbol{\zeta}$ and $\boldsymbol{\psi}$. For the
latent variables $\boldsymbol{\zeta}$, the regression coefficients
$\boldsymbol{\psi}_{\xi}$ for $\boldsymbol{\xi}$ and
$\boldsymbol{\beta}$ for the means of $\boldsymbol{\eta}$, and the
standard deviations $\boldsymbol{\sigma}$, these steps are similar
to the ones in \cite{kuhaetal22}, with adjustments to allow for the
facts that here $\boldsymbol{\eta}_{i}$ has four variables and that their
correlations vary by unit $i$. A description of these
steps is given in \ref{app:mcmc}. What is completely new here
is the procedure for sampling the coefficients $\boldsymbol{\alpha}$ of
the model (\ref{eta_corr_model}) for the conditional correlations of
$\boldsymbol{\eta}_{i}$. It is described in the rest of this section.

When $\boldsymbol{\alpha}$ is being sampled, all the other quantities
are taken as known and fixed at their most recently sampled values. The
latent variables $\boldsymbol{\eta}_{i}$ are thus also treated as
observed response variables in this model for their correlations. The
other parameters $\boldsymbol{\beta}$ and $\boldsymbol{\sigma}$ of the
distribution of $\boldsymbol{\eta}_{i}$ are also taken as known, and we
will omit them from the notation below. The posterior distribution that
we need is then written as
$p(\boldsymbol{\alpha}|\mathbf{X},\boldsymbol{\eta}) \,\propto\,
p(\boldsymbol{\eta}|\mathbf{X}; \boldsymbol{\alpha})\,
p(\boldsymbol{\alpha})$. As discussed in Section \ref{sub:Some Useful
Properties}, we consider this posterior over a convex and bounded set
$C_{\alpha,S_{XT}}$, where $S_{XT}$ is a finite test set of values for
$\boldsymbol{X}_{j}$ which also implies feasibility of correlations over
the full set $\mathbf{X}(S_{Z})$ that we are interested in. We specify a
joint uniform prior distribution $p(\boldsymbol\alpha) \, \propto\,
\mathbbm{1}(\boldsymbol{\alpha}\in C_{\alpha,S_{XT}})$ for
$\boldsymbol{\alpha}$ over this set.

Let $\alpha_{lm}$ denote any single element of $\boldsymbol{\alpha}$,
for $l=1,\dots,L$, $m=1,\dots,q$. The sampling algorithm updates one
$\alpha_{lm}$ at a time, taking all the other elements
$\boldsymbol{\alpha}_{-lm}$ fixed at their most recently sampled values.
Denote $\mathbf{R}_i(\alpha_{lm}) =
\mathbf{R}(\mathbf{X}_i;\alpha_{lm},\boldsymbol{\alpha}_{-lm})$ and
define the standardized residuals $\boldsymbol{\epsilon}=
[\boldsymbol{\epsilon}_1,\ldots,\boldsymbol{\epsilon}_n]^\top =
\mathbf{S}^{-1}\,(\boldsymbol{\eta}-\mathbf{X}\boldsymbol{\beta})$ where
$\boldsymbol{\eta}=[\boldsymbol{\eta}_{1},\dots,\boldsymbol{\eta}_{n}]^{\top}$
and $\mathbf{S}=\text{diag}(\sigma_{GP},\sigma_{RP},1,1)$. The
conditional posterior distribution from which $\alpha_{lm}$ should be drawn
is then
\begin{eqnarray}
    \lefteqn{p(\alpha_{lm}|\boldsymbol{\alpha}_{-lm},
    \boldsymbol{\epsilon}, \mathbf{X})
    \;\propto\; \prod_{i=1}^n \,
    p(\boldsymbol{\epsilon}_i|\boldsymbol\alpha,\mathbf{X})\,p(\alpha_{lm}|\boldsymbol{\alpha}_{-lm})}
    \nonumber\\
    &\propto& \prod_{i=1}^n \, |\mathbf{R}_i(\alpha_{lm})|^{-\frac{1}{2}}
    \, \exp\left( -\frac{1}{2}\boldsymbol{\epsilon}_i^\top
    \mathbf{R}_i(\alpha_{lm})^{-1} \boldsymbol{\epsilon}_i \right) \,
    \mathbbm{1}(a_{lm}< \alpha_{lm}< b_{lm}),
\label{eq:post_mu_dist}
\end{eqnarray}
where
$(a_{lm},b_{lm})$ is the range of $\alpha_{lm}$ in the subset of
$C_{\alpha,S_{XT}}$ given $\boldsymbol{\alpha}_{-lm}$. This involves $n$
matrix determinants and inverse operations, plus further determinants to
obtain the interval $(a_{lm},b_{lm})$ as described in Proposition
\ref{prop:conti_interval_mu} above. This would be computationally
demanding. However, these demands can be reduced because the sampling
updates only one parameter $\alpha_{lm}$ at a time. The calculation of the
determinant and inverse of $\mathbf{R}_i(\alpha_{lm})$ can be avoided
by maintaining and updating copies of working determinants and inverses.
These features are included in the general elementwise
Metropolis-Hastings (MH) procedure that we propose for sampling
$\boldsymbol\alpha$. It is given in Algorithm \ref{alg:alg1}, together with
Remarks 1--4 below. Assuming the resulting Markov chain satisfies the
standard regularity conditions
\citep[]{tierneyMarkovChainsExploring1994a,tierney1996introduction},
in which the detailed balance condition is met by construction,  \cref{alg:alg1} has the desired posterior
distribution as its unique stationary distribution.

\begin{algorithm}[!th]
  \SetAlgoLined
       1. \KwIn{
       Current parameters $\boldsymbol\alpha=(\alpha_{lm})$ for
       $l=1,\dots,L$, $m=1,\dots,q$. \\
       \hspace*{1em}For units $i=1,\dots,n$: Standardized residuals
       $\boldsymbol\epsilon_i=\mathbf{S}^{-1}\,(\boldsymbol{\eta}_i -
       \boldsymbol{\beta}^{\top}\mathbf{X}_i)$;\\
       \hspace*{2em}
       $\mathbf{R}_{i}^{-1}$ and $|\mathbf{R}_i|$ for correlation
       matrices $\mathbf{R}_{i}=\mathbf{R}(\mathbf{X}_{i};
       \boldsymbol{\alpha})$.\\
       \hspace*{1em}
       For a test set $S_{XT}=\{\mathbf{X}_{j}\,|\, j=1,\dots,T\}$: Upper triangular matrices
       $\mathbf\Gamma_j$ from\\ \hspace*{2em}the Cholesky decompositions
       $\mathbf{R}_{j}=\boldsymbol{\Gamma}_{j}^{\top}\boldsymbol{\Gamma}_{j}$
       of $\mathbf{R}_{j}=\mathbf{R}(\mathbf{X}_j;\boldsymbol{\alpha})$.
       }
2.  {\bf Metropolis-Hastings sampling:}

\For{$l = 1,\dots, L$}{

\For{$m = 1,\dots,q$}{

\paragraph{Proposal generation:}~

Calculate $(a_{lm},b_{lm})$ based on
$\mathbf{\Gamma}_1,\dots,\boldsymbol{\Gamma}_{T}$. See Remark 1
for more on this.

Generate $\alpha_{lm}'$ from a proposal distribution
$g(\alpha_{lm}'|\alpha_{lm})$.
See Remark 2 for more on how the proposal can be created.

\paragraph{Rejection:}~

Calculate $\mathbf{R}_i(\alpha_{lm}')^{-1}$ by updating
$\mathbf{R}_i(\alpha_{lm})^{-1}$ and
$|\mathbf{R}_i(\alpha_{lm}')|$ by updating
$|\mathbf{R}_i(\alpha_{lm})|$, for $i=1,...,n$; see Remark 3.

  Calculate the acceptance probability
    \[
      \pi(\alpha_{lm}\rightarrow \alpha_{lm}') =
      \min \Bigg\{1, \frac{p(\alpha_{lm}'|\boldsymbol{\alpha}_{-lm},
      \boldsymbol\epsilon,\mathbf{X})\,g(\alpha_{lm}|\alpha_{lm}')}
      {p(\alpha_{lm}|\boldsymbol{\alpha}_{-lm},
      \boldsymbol\epsilon,\mathbf{X})\, g(\alpha_{lm}'|\alpha_{lm})}\Bigg\}
    \]
    where $p(\alpha_{lm}|\boldsymbol{\alpha}_{-lm},
      \boldsymbol\epsilon,\mathbf{X})$ is given by equation
      (\ref{eq:post_mu_dist}).

    Sample $u\sim U(0,1)$.

    \If{$u > \pi(\alpha_{lm}\rightarrow \alpha_{lm}')$}{
      Reject $\alpha_{lm}'$\;
      \Continue
    }
    Accept $\alpha_{lm}'$ and update\\ \hspace*{1em}
    $\alpha_{lm}\rightarrow\alpha_{lm}'$,
    $\mathbf{R}_i(\alpha_{lm})^{-1}\rightarrow
    \mathbf{R}_i(\alpha_{lm}')^{-1}$, $|\mathbf{R}_i(\alpha_{lm})|
    \rightarrow |\mathbf{R}_i(\alpha_{lm}')|$.

    Update $\mathbf{\Gamma}_j(\alpha_{lm})\rightarrow
    \mathbf{\Gamma}_j(\alpha_{lm}')$ for $j=1,\dots,T$; see Remark 4.

}
}
      3. \KwOut{Updated $\boldsymbol\alpha$, $\mathbf{R}_{i}^{-1}$,
      $|\mathbf{R}_{i}|$ and $\boldsymbol{\Gamma}_{j}$.}
\caption{Elementwise Metropolis-Hastings procedure for sampling
$\boldsymbol{\alpha}$}
\label{alg:alg1}
\end{algorithm}

\vspace*{1ex}

\textbf{Remark 1}: \emph{Calculating feasible interval for $\alpha_{lm}$ by
Cholesky decomposition.}

\cref{prop:conti_interval_mu} and Lemma
\ref{lemma:conti_rho} in \ref{app:proofs_of_lemmas} descibe one
way of calculating the interval $(a_{lm},b_{lm})$. This requires the
calculation of $3T$ determinants of correlation matrices. An
alternative, more efficient procedure for its first steps can be
obtained by adapting a method proposed by
\cite{wongEfficientEstimationCovariance2003}.
Let $\mathbf{R}_j =\mathbf{\Gamma}_j^\top\mathbf{\Gamma}_j$ as
defined in the Input statement of \cref{alg:alg1}, where
$\boldsymbol{\Gamma}_{j}=(\gamma_{k_{1},k_{2}}^{(j)})$. Recall that $K$ denotes
the dimension of $\mathbf{R}_{j}$, and assume that $\rho_{l}$ in
Lemma \ref{lemma:conti_rho} corresponds to the $(K,K-1)$th element of
$\mathbf{R}_j$.
We then have
$g_{jl} = \sum_{k=1}^{K-2} \gamma_{k,K-1}^{(j)}\gamma_{k,K}^{(j)}$ and
    $h_{jl} = \gamma_{K-1,K-1}^{(j)}
    (1-\sum_{k=1}^{K-2}(\gamma_{k,K}^{(j)})^{2})^{1/2}$, and
$(a_{lm},b_{lm})$ can be obtained by plugging in $g_{jl}$
and $h_{jl}$ into (\ref{eq:interval_j}) in
\cref{prop:conti_interval_mu} as before.
If $\rho_l$ is not the $(K,K-1)$th element of
$\mathbf{R}_j$, we can permute the indices
with a permutation matrix $\mathbf{P}$ so that it is
the $(K,K-1)$th element of the matrix $\mathbf{P}^\top \mathbf{R}_j
\mathbf{P} = (\mathbf{\Gamma}_j
\mathbf{P})^\top(\mathbf{\Gamma}_j \mathbf{P})$, followed by
a Givens rotation by an orthogonal matrix $\mathbf{Q}$ such that
$\mathbf{Q}\mathbf{\Gamma}_j \mathbf{P} = \tilde{\mathbf{\Gamma}}_j$,
where $\tilde{\mathbf{\Gamma}}_j$ is upper-triangular and
$\mathbf{P}^\top \mathbf{R}_j \mathbf{P} =
\tilde{\mathbf{\Gamma}}_j^\top \tilde{\mathbf{\Gamma}}_j$. Then apply
the calculation above to $\tilde{\boldsymbol{\Gamma}}_{j}$.

\vspace*{1em}

\textbf{Remark 2}: \emph{Generating proposal values for $\alpha_{lm}$}.

We have used a simple random walk Metropolis sampler. It generates
the proposal through an independent Gaussian increment to the
previous value, as $\alpha_{lm}^\prime = \alpha_{lm} + \gamma_m \delta$,
where $\delta$ is drawn from the standard normal distribution.
Thus $\alpha_{lm}|\alpha_{lm}'\sim N(\alpha_{lm}', \gamma_{m}^{2})$.
The step size $\gamma_{m}$ should be chosen to achieve a good balance
between rejection rate and mixing efficiency. We have used
$\gamma_m = C(\sqrt{n}||(X_{1m},\dots,X_{nm})^{\top}||_{\infty})^{-1}$,
 where $||\cdot||_{\infty}$ is the infinity norm and $C$ is a
chosen constant, the same for all $\gamma_m$, which is used to control
rejection rates in the range 0.7--0.8.
The sampler was efficient enough in our real
data analysis when the step sizes were chosen appropriately.

An alternative would be to use the ARMS algorithm
\citep{gilksAdaptiveRejectionMetropolis1995} to adaptively construct the
proposal function of $\alpha_{lm}$ in $(a_{lm},b_{lm})$. This can
improve the acceptance rate but the algorithm may require the likelihood
function $p(\boldsymbol\epsilon|\boldsymbol\alpha,\mathbf{X})$ to be
evaluated multiple times based on the rejection condition, whereas in
the random walk Metropolis method it needs to be calculated at most once
in each iteration. Other methods for improving the acceptance rate exist
\citep[]{chibAnalysisMultivariateProbit1998}, but their implementation
is more complex and relies heavily on tuning.

\vspace*{1em}

\textbf{Remark 3}: \emph{Updating the determinant and inverse of
correlation matrix}.

Here we want to update the determinant and inverse of
$\mathbf{R}_{i}(\alpha_{lm})$ to those of
$\mathbf{R}_{i}(\alpha_{lm}')$.
Suppose that the correlation parameter $\rho_{l}$ corresponds to the
$(k_1,k_2)$th
element of $\mathbf{R}_i(\alpha_{lm})$. Let
$\varepsilon_{ilm} = (\alpha_{lm}' -
\alpha_{lm})X_{im}$, and denote by $\mathbf{w}_{1}$ and $\mathbf{w}_{2}$
the $K \times 1$ vectors which are zero except that the $k_1$th element
of $\mathbf{w}_{1}$
and the $k_2$th element of $\mathbf{w}_{2}$ are
$\sqrt{|\varepsilon_{ilm}|}$.
Then
\[
\mathbf{R}_i(\alpha_{lm}') =
\left[\mathbf{R}_i(\alpha_{lm}) +(\text{sgn}(\varepsilon_{ilm})
\mathbf{w}_{1})\mathbf{w}_{2}^\top\right] + \mathbf{w}_{2}(\text{sgn}(\varepsilon_{ilm})\mathbf{w}_{1})^\top.
\]
Since this is of the form
$(\mathbf{A}+\mathbf{u}\mathbf{v}^{\top})+\mathbf{v}\mathbf{u}^{\top}$,
$\mathbf{R}_i(\alpha_{lm}')^{-1}$ can be computed efficiently with two
rank-1 updates by applying twice the Sherman-Morrison formula

  \[
    (\mathbf{A} + \mathbf{u}\mathbf{v}^\top)^{-1} =
    \mathbf{A}^{-1}
    - \frac{\mathbf{A}^{-1}\mathbf{u}\mathbf{v}^\top
    \mathbf{A}^{-1}}{1 + \mathbf{v}^\top \mathbf{A}^{-1}\mathbf{u}}
  \]
and $|\mathbf{R}_i(\alpha_{lm}')|$ can be calculated by updating
$|\mathbf{R}_i(\alpha_{lm})|$ through two applications of
  $
    |\mathbf{A} + \mathbf{u}\mathbf{v}^\top| =
    (1+\mathbf{v}^\top
    \mathbf{A}^{-1}\mathbf{u})\,
    |\mathbf{A}|,
  $
the second of which employs the first update of the inverse.
These steps reduce the computation complexity of
$p(\boldsymbol{\epsilon}_{i}|\boldsymbol{\alpha},\mathbf{X})$
from $O(K^3)$ to $O(K^2)$.

\vspace*{1em}

\textbf{Remark 4}: \emph{Updating the Cholesky decomposition of a
correlation matrix}.

Let $\varepsilon_{jlm} = (\alpha_{lm}' - \alpha_{lm})X_{jm}$. Let
$\mathbf{w}_{1}$ and
$\mathbf{w}_{2}$ be defined as in Remark 3, and define $\mathbf{w}$ as
the $K\times 1$ vector where the $k_1$th and $k_2$th
elements are $\sqrt{|\varepsilon_{jlm}|}$ and the other elements are
zero. Then we can write
\[
\mathbf{R}_j(\alpha_{lm}') = \left[\left(\mathbf{R}_j(\alpha_{lm})
+\text{sgn}(\varepsilon_{jlm}) \mathbf{w}\mathbf{w}^\top\right) -
\text{sgn}(\varepsilon_{jlm}) \mathbf{w}_{1}\mathbf{w}_{1}^\top\right]
 -
\text{sgn}(\varepsilon_{jlm})\mathbf{w}_2\mathbf{w}_2^\top.
\]
The Cholesky decomposition of $\mathbf{R}_j(\alpha_{lm}')$ can be
computed efficiently from this, with three rank-1 updates for the Cholesky decomposition of the
form $\mathbf{A}+ \mathbf{u}\mathbf{u}^{\top}$ or $\mathbf{A} - \mathbf{u}\mathbf{u}^{\top}$\citep[]{Seeger:161468}; built-in
functions for this are available in Matlab and linear algebra libraries
like Eigen \citep[]{eigenweb}. This updating rule reduces the computation complexity of the Cholesky decomposition $\mathbf{R}_j=\mathbf{\Gamma}_j^\top\mathbf{\Gamma}_j$ from $O(K^3)$ to $O(K^2).$

\vspace*{1ex}

Alternatives to Algorithm \ref{alg:alg1} could also be considered. In
 cases when $\mathbf{X}=\mathbf{X}(\mathbf{Z})$ is a complex function
 such as a cubic spline, for better efficiency the element-wise MH
 algorithm could be replaced with a blockwise algorithm where subvectors
 of $\boldsymbol\alpha$ can be proposed and rejected together. Apart
 from the MH algorithm we use in this paper, we note that the ``griddy
 Gibbs'' sampler discussed in \cite{johnbar2000} also works here in
 principle, where the feasible intervals for each $\alpha_{lm}$ can be
 discretized into grids. However, the computational efficiency for
 evaluating posterior function over these grids may suffer.

\section{Analysis of child-parent exchanges of support} 
\label{sec:analysis_of_the_help_exchanges}

\subsection{Introduction and research questions}
\label{subsec:background_and_RQs}

The model defined in Section \ref{sec:latent_variable_model} was fitted
to the UKHLS data on exchanges of support between respondents and their
non-coresident parents that were introduced in Section
\ref{sec:data_on_help_exchanges}, using the method of estimation that
was described in Section \ref{sec:estimation_of_the_latent_variable_model}. Here
receiving and giving help are modelled jointly, treating
practical and financial support as distinct but correlated outcomes. We
investigate the following research questions:
\begin{enumerate}[(a)]
\item What individual characteristics are associated with
higher or lower levels of giving practical and financial help to the
parents, and receiving such help from the parents?\\[-4ex]
\item To what extent are exchanges reciprocated and how does reciprocity
vary according to individual characteristics?\\[-4ex]
\item Are practical and financial support substitutes for one another or
are they complementary, and how does this depend on individual
characteristics?
\end{enumerate}
Questions (b) and (c) refer to within-person correlations between the
helping tendencies. For (b), higher levels of reciprocity would
correspond to positive correlations between giving and receiving help.
For (c), positive correlations between the tendencies to give (or to
receive) practical and financial help would suggest that the two types
of support are complementary (i.e.\ given together), and negative
correlations that they are substitutes.

Estimates $\tilde{\boldsymbol{\phi}}$ of the parameters of the
measurement model were obtained first, as explained in Section
\ref{sub:estimation_of_the_measurement_model}. They are shown in Table
S1 of the supplementary materials. The loading parameters are positive,
so the latent variables $\eta_{GP}$ and $\eta_{RP}$ are defined so that
larger values of them imply higher tendencies to give and receive
practical help (and the same is true by construction for the financial
help variables $\eta_{GF}$ and $\eta_{RF}$). The measurement parameters
were then fixed at $\tilde{\boldsymbol{\phi}}$ in the estimation of the
rest of the model below.

The structural model for the joint distribution of the latent variables
was estimated using the MCMC algorithm described in Section
\ref{sub:estimation_of_the_structural_model2} and \ref{app:mcmc}.
Estimated parameters and some predicted values for these models are
shown in Tables \ref{t_model_eta1}-\ref{t_model_corr_pred} and in Tables
S2--S3 of the supplementary materials. They are based on a sample of
380,000 draws of the parameters $\boldsymbol{\psi}$, obtained by pooling
two MCMC chains of 200,000 iterations each, with a burn-in sample of
10,000 omitted from each chain.  Convergence was assessed by visual
inspection of trace plots of the two chains which suggested adequate
mixing. In the role of the set of interest $S_{Z}$ for the covariates
(as defined in Section \ref{sub:Some Useful Properties}), we used the
simple choice of the set of all the $n$ observed values of
$\mathbf{Z}_i$ in the data, and as the test set $S_{XT}$ all the
distinct values of $\mathbf{X}_{i}=\mathbf{X}(\mathbf{Z}_{i})$ implied
by them.

Estimated parameters of the multinomial logistic model
(\ref{eq:xi_dist}) for the joint distribution of the binary latent class
variables $(\xi_G,\xi_R)$ are shown in Table S2, and fitted class
probabilities $p(\xi_G = 1)$ and $p(\xi_R = 1)$ from it given different
values of the covariates in Table~S3. This model component is included
primarily to allow for zero-inflation in the observed item responses, so
it is not our main focus. We could, however, also interpret the classes
defined by $\xi_G = 1$ and $\xi_R = 1$ as latent sub-populations of
`givers' and `receivers' of help respectively. The estimated overall
proportions of these classes, averaged over the sample distribution of
the covariates, are 0.67 for `givers' and 0.62 for `receivers'.

The focus of interest is the model for the joint distribution of
$\boldsymbol{\eta}=(\eta_{GP},\eta_{GF},\eta_{RP},\eta_{RF})^{\top}$,
which we interpret as continuous latent tendencies for the adult
respondents to give practical and financial help to and receive help
from their non-coresident parents, after accounting for the
zero-inflation. We consider first results for the linear model
(\ref{eta_mean}) for the means of $\boldsymbol{\eta}$, which is used to
answer research question (a), and then discuss estimates of the model
(\ref{eta_corr})--(\ref{eta_corr_model}) for their correlations,
corresponding to questions (b) and (c).

\subsection{Predictors of levels of giving and receiving help}
\label{subsec:predictors_of_xi_and_eta}

Table \ref{t_model_eta1} shows the estimated coefficients of the
predictors of the means of practical ($\eta_{GP}$) and financial
($\eta_{GF}$) help given by respondents to parents. There is little
evidence that the respondent's partnership status or the presence or age
of their children are associated with the tendency to give help. Women
tend to give more practical help than men, but there is no gender
difference in giving financial help. Indicators of lower socioeconomic
status or a more difficult economic situation of the respondent (lower
education, not being a homeowner, lower household income, and not being
employed) are associated with a higher tendency to give practical help,
while having more education and higher household income predict a higher
tendency to give financial help. These results are consistent with a
pattern where children give help to the best of their ability, with less
well-off children giving, on average, relatively more practical support
and less financial support (this does not, however, tell us about
possible substitution of types of help by the same person; for that, we
will turn to the within-person correlations in the next section).
However, the results for household tenure and employment status (where
home owners and the employed also tend to give less financial help)
deviate from this pattern, after controlling for education and income.
There is also some evidence that respondents with one sibling give less
help than those with none, which could suggest some sharing of support
between the siblings (although there is no similar reduction for those
with more siblings).

Having a parent who lives alone and older parental age are positively
associated with giving both forms of help, with the positive association
between age and financial help emerging when the oldest parent reaches
their early 70s. Both of these findings are consistent with children
giving help according to parental need. After controlling for  parental
age, the (correlated) respondent's age has an inverse U-shaped
relationship with giving both practical and financial help, with highest
levels of giving at around ages 43 and 49 respectively. Finally,
respondents who live more than an hour away from the nearest parent have
a lower tendency to give practical help, but a higher tendency to give
financial help. As for the effects of socioeconomic status, the
different directions of these associations suggest differences in the
mix of different types of help related to the giver's circumstances, in
this case according to how feasible it is to provide practical help.

Covariate effects on levels of practical and financial help that the
respondents receive from their parents (variables $\eta_{RP}$ and
$\eta_{RF}$) are shown in Table \ref{t_model_eta2}. Women tend to
receive more of both types of support than men. Expected levels of
support from parents are also higher for respondents who are not
employed, have less education, or have no coresident partner, all of
which can be taken to indicate higher levels of need for support.
Respondents with two or more siblings tend to receive less of either
form of help than those from one or two-child families, which may
reflect greater competition for parental resources in larger families.
For financial help, the tendency to receive such help is higher for
respondents who have lower household income or who rent rather than own
their homes, as well as for those with very young or secondary school
age children. These associations are also consistent with parents
providing financial assistance to children who are most in need.

Levels of both practical and financial help received decline with the
respondent's age, which is consistent with reduced need by respondents.
As a function of the oldest parent's age, receipt of practical help also
declines from age 67 onwards, but the tendency to receive financial help
increases with parental age. This may be interpreted as another instance
of the balance of different types of help depending on the giver's
capacities, in this case with older parents being more able to give
financial than practical support. Finally, longer travel time between
the respondent and their nearest parent is associated with less
practical and more financial help, as it was also for help from
respondents to parents.

\subsection{Models for the correlations: Predictors of symmetry in
exchanges and complementarity of practical and financial help}
\label{subsec:predictors_of_correlations}

Estimated coefficients ($\hat{\boldsymbol{\alpha}}$) of the model
(\ref{eta_corr})--(\ref{eta_corr_model}) for the residual correlations
of $\boldsymbol{\eta}= (\eta_{GP},\eta_{GF},\eta_{RP},\eta_{RF})$ are
shown in Table \ref{t_model_corr_coeff} and some fitted correlations
from these models in Table \ref{t_model_corr_pred}. Here we included as
covariates the respondent's age, age squared, gender, household income,
and travel time to the nearest parent. Whereas the models in Section
\ref{subsec:predictors_of_xi_and_eta} concern the mean of each helping
tendency separately, these correlations focus on their joint
distribution for a given child-parent dyad, over and above the levels
predicted by the mean models. They can be used to investigate research
questions (b) and (c) above.

The four correlations between the tendencies to give and receive help
(of the same or different type) can be viewed as measures of reciprocity
or symmetry in exchanges between children and their parents (question
b). Results for them are given in the first four columns of each table,
where for ease of interpretation we focus on the fitted values in
Table~\ref{t_model_corr_pred}. Consider first the correlations averaged
over the sample distribution of the covariates, shown on the first row.
There is a moderate positive correlation of 0.38 between giving and
receiving practical help (GP $\leftrightarrow$ RP). In other words, when
a child has a high tendency to give practical help to their parent(s),
relative to what would be predicted by their own and the parents'
characteristics, they also tend to receive a higher-than-average level
of support from the parents. This suggests a fair amount of reciprocity
in practical help. The other three correlations are weaker, indicating
little dyad-level reciprocity in anything other than practical help.
What is not observed here are any substantial negative correlations.
They would indicate that when the tendency to help is high in one
direction it is low in the other, as would happen for example if help
was given only in the direction of greater need. This is not seen here
even for giving and receiving financial help, even though we might have
expected financial exchanges to be largely unidirectional. One possible
explanation of this is that the single financial support item covers
also small sums of money, which may be exchanged more frequently and
symmetrically than large ones.

The (GP $\leftrightarrow$ RP) correlation is also the one for which we
see the most noticeable covariate effects, as illustrated by the other
rows of Table \ref{t_model_corr_pred}. It declines sharply with age, and
it is significantly higher for men than for women and among parents and
children who live farther apart. Reciprocity in practical support is
highest at younger ages of the adult children. We note that this
captures a different aspect of the effects of age than the mean models
in Section \ref{subsec:predictors_of_xi_and_eta}. There respondent's age
was negatively associated with tendency to give practical help and
positively associated (up to age around 43) with tendency to receive it.
Thus younger individuals tend to give less practical help and receive
more of it, and the expected balance of support is more toward help from
parents to children, than is the case at other ages (comparable
conclusions were reached in a different way by \citealt{mudrazija:16},
who considered net financial values of the differences between these two
directions). The residual correlations show that, around these expected
levels, for younger respondents the level of practical help that they do
(or do not) give is particularly strongly predictive of how much support
they receive. Similarly, the gender difference in the correlation
suggests that men are more likely than women to engage in two-way
exchanges or not exchange practical help at all.

The only other clearly significant covariate effects on the correlations
that relate to reciprocity are those between within-dyad distance and
the (GP $\leftrightarrow$ RF), (GF $\leftrightarrow$ RF) and (GP
$\leftrightarrow$ RP) correlations. Recall that the models for the means
showed that the balance of the expected levels of different types of
help moves towards more financial and less practical support when the
child and the parent(s) live further apart. Of the residual correlations
here, (GP $\leftrightarrow$ RP) is quite strongly positive when the
distance is longer vs.\ less positive when it is shorter, while (GP
$\leftrightarrow$ RF) is near zero vs.\ moderately positive and (GF
$\leftrightarrow$ RF) moderately negative vs.\ near zero similarly (and
GF $\leftrightarrow$ RP is always small). One possible interpretation of
these different patterns is that among children and parents who live
further apart providing practical support requires a greater effort and
the tendency to give such support may be higher when reciprocated. For
such dyads, financial help may also more often involve one-way (and
perhaps larger) transfers which are less often (and less easily)
reciprocated by practical help.

The two remaining correlations, between the tendencies to give financial
and practical help and between the tendencies to receive financial and
practical help, are used to examine whether one form of help that a
person may give serves as a substitute for the other or whether they are
complementary, and whether this varies according to individual
characteristics (research question c). Here the mean models in Section
\ref{subsec:predictors_of_xi_and_eta} also give information about one
version of this question, when they show that the expected balance of
the two types of help is, on average, different for dyads with different
characteristics. This is most obvious when the coefficient of a
covariate has different signs for practical and financial help, as it
does for example for the distance between respondent and their parents
(a similar result for expected levels of financial vs.\ time assistance
given distance was found by \citealt{bonsang:07} in a cross-national
European study). However, this is again not the same as the question of
substitution for a person, i.e.\ whether the level of one kind of help
that he or she gives predicts higher or lower levels of the other kind
of help.

Results for the correlations that address this question are given in the
final two columns of Tables \ref{t_model_corr_coeff} and
\ref{t_model_corr_pred}. The fitted correlations are positive overall
and in all sub-groups defined by the covariates. This indicates clearly
that within a person the types of help are not supplementary but
complementary: a person (child or parent) who has a high tendency to
give one kind of help (relative to what would be expected given the
characteristics of their dyad) also has a high tendency to give the
other kind of help. The most noticeable covariate effect that holds for
both children and the parents is that the degree of complementarity in
practical and financial help is greater when the child-parent distance
is small. For help received from the parents, complementarity also
declines with the respondent's (and thus in effect also the parents')
age. This suggests that at older ages the parents more often tend to
limit the support that they give to one of these types (mostly likely
financial help, in light of the results in Table \ref{t_model_eta2})
rather than both of them.

In conclusion, we return to the research questions that were stated in
Section \ref{subsec:background_and_RQs}. The first question was
addressed by the models for the mean levels of helping tendencies in
Section \ref{subsec:predictors_of_xi_and_eta}. Their results may be
summarised in terms of two broad types of characteristics: the
capacities of a giver of support and the level of need of the recipient.
The model results indicate clearly that recipients with higher level of
need (such as children with less privileged socioeconomic status or
parents who are older or live alone) tend to receive more support. For
capacities of giving, the results are more subtle. There is no strong
evidence that lower capacity is associated with less help given in some
overall sense. Instead, different types of individuals tend to give the
types of help that they are best able to give, e.g.\ with less wealthy
children giving relatively more practical than financial help to their
parents, and older parents providing relatively more financial help to
their children.

The other two research questions correspond to the models for residual
correlations in this section. The results show evidence of reciprocity
between children and parents in practical help, and of within-person
complementarity in giving different types of help. A prominent covariate
effect on these correlations was found for the distance between children
and their parents, with different patterns of correlations between
helping tendencies of different types and directions for children who
lived far from rather than close to their parents.

\begin{table}[!htbp]
\centering
\caption{
Estimated parameters of the linear model for the expected value of the tendency to
give practical help ($\eta_{GP}$) and to give financial help
($\eta_{GF}$) to individuals' non-coresident parents.
The estimates are posterior means from MCMC samples (with posterior standard
deviations in parentheses).
}

\vspace*{1ex}
\label{t_model_eta1}
\begin{tabular}{lrrcrr}
\hline
&
\multicolumn{2}{l}{Giving} &&
\multicolumn{2}{l}{Giving} \\
&
\multicolumn{2}{l}{practical help} &&
\multicolumn{2}{l}{financial help} \\
&  Estimate & (s.d.) & & Estimate & (s.d.)\\
\hline
\rule{0pt}{3ex}\emph{Estimated coefficients:}
& $\hat{\boldsymbol{\beta}}_{GP}$ &&&
 $\hat{\boldsymbol{\beta}}_{GF}$ &
\\[1ex]
Intercept &  $-0.70^{***}$ & $(0.18)$  &&  $-2.35^{***}$ & $(0.31)$ \\[.5ex]
\multicolumn{2}{l}{\textbf{Respondent (child) characteristics}} & &\\
Age$^{\dagger}$ ($\times 10$)
& $0.03^{\phantom{***}}$ & $(0.03)$  &&  $0.12^{***}$ & $(0.04)$ \\[.5ex]
Age squared$^{\dagger}$ ($\times 10^3$) & $-0.60^{***}$ & $(0.12)$  &&  $-0.70^{***}$ & $(0.19)$ \\[.5ex]
Gender    &                &&                &           & \\
~~Female (vs.\ Male) & $0.41^{***}$ & $(0.03)$  &&  $0.03^{\phantom{***}}$ & $(0.04)$ \\[.5ex]
Partnership status              &               &   &           & \\
~~Partnered (vs.\ Single) & $-0.04^{\phantom{***}}$ & $(0.03)$  &&  $0.01^{\phantom{***}}$ & $(0.05)$ \\[.5ex]
\multicolumn{6}{l}{Age of youngest coresident child (vs.\ No children):}\\
~~0--1 years    & $-0.08^{\phantom{***}}$ & $(0.06)$  &&  $-0.05^{\phantom{***}}$ & $(0.09)$ \\
~~2--4 years    & $0.01^{\phantom{***}}$ & $(0.05)$  &&  $0.03^{\phantom{***}}$ & $(0.08)$ \\
~~5--10 years   & $0.02^{\phantom{***}}$ & $(0.04)$  &&  $0.09^{\phantom{***}}$ & $(0.07)$ \\
~~11--16 years  & $-0.04^{\phantom{***}}$ & $(0.05)$  &&  $-0.10^{\phantom{***}}$ & $(0.07)$ \\
~~17-- years   & $0.03^{\phantom{***}}$ & $(0.04)$  &&  $-0.03^{\phantom{***}}$ & $(0.06)$ \\[.5ex]
Number of siblings (vs.\ None) & &  \\
~~1 & $-0.08^{*\phantom{**}}$ & $(0.04)$  &&  $-0.12^{*\phantom{**}}$ & $(0.07)$ \\
~~2 or more     & $0.00^{\phantom{***}}$ & $(0.04)$  &&  $0.06^{\phantom{***}}$ & $(0.07)$ \\[.5ex]
Longstanding illness (vs.\ No) & $0.07^{*\phantom{**}}$ & $(0.04)$  &&  $0.07^{\phantom{***}}$ & $(0.06)$ \\[.5ex]
\multicolumn{2}{l}{Employment status (vs.\ Employed)} &           &           & \\
~~Not employed & $0.21^{***}$ & $(0.03)$  &&  $0.11^{**\phantom{*}}$ & $(0.05)$ \\[.5ex]
Education (vs.\ Secondary or less)
& &           &           & \\
~~Post-secondary  & $-0.05^{**\phantom{*}}$ & $(0.03)$  &&  $0.12^{***}$ & $(0.04)$ \\[.5ex]
Household tenure (vs.\ Renter) &&           &           & \\
~~Own home outright or with mortgage & $-0.17^{***}$ & $(0.03)$  &&
$-0.19^{***}$ & $(0.05)$ \\[.5ex]
Logarithm of household equivalised income& $-0.04^{**\phantom{*}}$ & $(0.02)$  &&  $0.09^{***}$ & $(0.03)$ \\[8pt]
\multicolumn{3}{l}{\textbf{Parent characteristics}}         &        &\\
Age of the oldest living parent$^{\dagger}$ ($\times 10$)
& $0.28^{***}$ & $(0.02)$  &&  $-0.02^{\phantom{***}}$ & $(0.04)$ \\[.5ex]
Age of the oldest parent squared$^{\dagger}$ ($\times 10^{3}$)  & $0.52^{***}$ & $(0.11)$  &&  $0.63^{***}$ & $(0.17)$ \\[.5ex]
At least one parent lives alone (vs.\ No) & $0.33^{***}$ & $(0.03)$  &&
$0.24^{***}$ & $(0.04)$ \\[8pt]
\multicolumn{3}{l}{\textbf{Child-parent characteristics}} &&& \\
Travel time to the nearest parent &&&&\\
~~More than 1 hour (vs.\ 1 hour or less)  & $-0.43^{***}$ & $(0.04)$  &&  $0.14^{**\phantom{*}}$ & $(0.05)$ \\[8pt]
\emph{Residual s.d.:}
& $\hat{\sigma}_{GP}$ &&&&\\
& $0.73^{\phantom{***}}$ &$(0.01)$ && 1 &\\
\hline
\multicolumn{6}{l}{\small The posterior credible interval excludes zero at level 90\% (*), 95\% (**) or 99\% (***).}\\
\multicolumn{6}{l}{\small $\dagger$ Age of respondent is centered at 40, and age of oldest living parent at 70.}
\end{tabular}
\end{table}

\begin{table}[!htbp]
\centering
\caption{
Estimated parameters of the linear model for the expected value of the tendency to
receive practical help ($\eta_{RP}$) and to receive financial help
($\eta_{RF}$) from individuals' non-coresident parents.
The estimates are posterior means from MCMC samples (with posterior standard
deviations in parentheses).
}

\vspace*{1ex}
\label{t_model_eta2}
\begin{tabular}{lrrcrr}
\hline
&
\multicolumn{2}{l}{Receiving} &&
\multicolumn{2}{l}{Receiving} \\
&
\multicolumn{2}{l}{practical help} &&
\multicolumn{2}{l}{financial help} \\
&  Estimate & (s.d.) & & Estimate & (s.d.)\\
\hline
\rule{0pt}{3ex}\emph{Estimated coefficients:}
& $\hat{\boldsymbol{\beta}}_{RP}$ &&&
 $\hat{\boldsymbol{\beta}}_{RF}$ &
\\[1ex]
Intercept &  $-2.17^{***}$ & $(0.23)$  &&  $1.03^{***}$ & $(0.34)$ \\[.5ex]
\multicolumn{2}{l}{\textbf{Respondent (child) characteristics}} & &\\
Age$^{\dagger}$ ($\times 10$)
& $-0.26^{***}$ & $(0.03)$  &&  $-0.28^{***}$ & $(0.04)$ \\[.5ex]
Age squared$^{\dagger}$ ($\times 10^3$)  & $-0.16^{\phantom{***}}$ & $(0.18)$  &&  $-0.46^{*\phantom{**}}$ & $(0.23)$ \\[.5ex]
Gender    &                &&                &           & \\
~~Female (vs.\ Male) & $0.27^{***}$ & $(0.03)$  &&  $0.15^{***}$ & $(0.04)$ \\[.5ex]
Partnership status              &               &   &           & \\
~~Partnered (vs.\ Single) & $-0.35^{***}$ & $(0.04)$  &&  $-0.30^{***}$ & $(0.05)$ \\[.5ex]
\multicolumn{6}{l}{Age of youngest coresident child (vs.\ No children):}\\
~~0--1 years    & $0.02^{\phantom{***}}$ & $(0.05)$  &&  $0.14^{*\phantom{**}}$ & $(0.07)$ \\
~~2--4 years    & $-0.03^{\phantom{***}}$ & $(0.05)$  &&  $0.07^{\phantom{***}}$ & $(0.06)$ \\
~~5--10 years   & $-0.09^{**\phantom{*}}$ & $(0.04)$  &&  $-0.02^{\phantom{***}}$ & $(0.06)$ \\
~~11--16 years  & $-0.11^{*\phantom{**}}$ & $(0.06)$  &&  $0.18^{**\phantom{*}}$ & $(0.07)$ \\
~~17-- years   & $-0.12^{\phantom{***}}$ & $(0.07)$  &&  $0.02^{\phantom{***}}$ & $(0.09)$ \\[.5ex]
Number of siblings (vs.\ None) & &  \\
~~1 & $0.00^{\phantom{***}}$ & $(0.05)$  &&  $-0.07^{\phantom{***}}$ & $(0.07)$ \\
~~2 or more     & $-0.14^{***}$ & $(0.05)$  &&  $-0.25^{***}$ & $(0.06)$ \\[.5ex]
Longstanding illness (vs.\ No) & $0.03^{\phantom{***}}$ & $(0.05)$  &&  $0.06^{\phantom{***}}$ & $(0.06)$ \\[.5ex]
\multicolumn{2}{l}{Employment status (vs.\ Employed)} &           &           & \\
~~Not employed & $0.21^{***}$ & $(0.04)$  &&  $0.14^{**\phantom{*}}$ & $(0.05)$ \\[.5ex]
Education
(vs.\ Secondary or less)
& &           &           & \\
~~Post-secondary  & $-0.06^{**\phantom{*}}$ & $(0.03)$  &&  $-0.06^{\phantom{***}}$ & $(0.04)$ \\[.5ex]
Household tenure (vs.\ Renter) &&           &           & \\
~~Own home outright or with mortgage & $0.08^{**\phantom{*}}$ & $(0.03)$
&&  $-0.34^{***}$ & $(0.05)$ \\[.5ex]
Logarithm of household equivalised income& $0.01^{\phantom{***}}$ & $(0.02)$  &&  $-0.14^{***}$ & $(0.03)$ \\[8pt]
\multicolumn{3}{l}{\textbf{Parent characteristics}}         &        &\\
Age of the oldest living parent$^{\dagger}$ ($\times 10$)
& $-0.03^{\phantom{***}}$ & $(0.03)$  &&  $0.22^{***}$ & $(0.04)$ \\[.5ex]
Age of the oldest parent squared$^{\dagger}$ ($\times 10^{3}$)  & $-0.44^{***}$ & $(0.15)$  &&  $0.28^{\phantom{***}}$ & $(0.19)$ \\[.5ex]
At least one parent lives alone (vs.\ No) & $-0.05^{\phantom{***}}$ &
$(0.03)$  &&  $0.04^{\phantom{***}}$ & $(0.04)$ \\[8pt]
\multicolumn{3}{l}{\textbf{Child-parent characteristics}} &&& \\
Travel time to the nearest parent &&&&\\
~~More than 1 hour (vs.\ 1 hour or less)  & $-0.42^{***}$ & $(0.05)$  &&  $0.27^{***}$ & $(0.06)$ \\[8pt]
\emph{Residual s.d.:}
& $\hat{\sigma}_{RP}$ &&&&\\
& $0.68^{\phantom{***}}$ &$(0.02)$ && 1&\\
\hline
\multicolumn{6}{l}{\small The posterior credible interval excludes zero at level 90\% (*), 95\% (**) or 99\% (***).}\\
\multicolumn{6}{l}{\small $\dagger$ Age of respondent is centered at 40, and age of oldest living parent at 70.}
\end{tabular}
\end{table}

\begin{table}[!htbp]
  \centering
\caption{Estimated coefficients ($\hat{\boldsymbol{\alpha}}$) of the
model for the residual correlations of the tendencies to give and receive practical help (GP and RP)
and to give and receive financial help (GF and RF).
The estimates are posterior means from MCMC samples (with posterior standard
deviations in parentheses).}
\vspace*{1ex}
\label{t_model_corr_coeff}
  \small
      \begin{tabular}{lrrrr|rr}
      \hline
      & \multicolumn{6}{c}{Correlation}\\
       & \multicolumn{1}{c}{GP$\leftrightarrow$RP} & \multicolumn{1}{c}{GP$\leftrightarrow$RF} & \multicolumn{1}{c}{GF$\leftrightarrow$RP} & \multicolumn{1}{c}{GF$\leftrightarrow$RF} & \multicolumn{1}{c}{GP$\leftrightarrow$GF} & \multicolumn{1}{c}{RP$\leftrightarrow$RF}\\
      \hline
      \rule{0pt}{3ex}Intercept  & $0.087^{\phantom{***}}$ & $0.166^{\phantom{***}}$ & $-0.133^{\phantom{***}}$ & $-0.126^{\phantom{***}}$ & $0.475^{***}$ & $0.148^{\phantom{***}}$\\
      & $(0.171)^{\phantom{**}}$ & $(0.186)^{\phantom{**}}$ & $(0.220)^{\phantom{**}}$ & $(0.159)^{\phantom{**}}$ & $(0.174)^{\phantom{**}}$ & $(0.197)^{\phantom{**}}$\\ [8pt]
     Age of respondent$^{\dagger}$ & $-0.014^{***}$ & $0.004^{*\phantom{**}}$ & $0.003^{\phantom{***}}$ & $-0.001^{\phantom{***}}$ & $-0.002^{\phantom{***}}$ & $-0.009^{***}$\\
      & $(0.002)^{\phantom{**}}$ & $(0.002)^{\phantom{**}}$ & $(0.003)^{\phantom{**}}$ & $(0.003)^{\phantom{**}}$ & $(0.002)^{\phantom{**}}$ & $(0.002)^{\phantom{**}}$\\ [8pt]
     Age squared$^{\dagger}$ ($\times 10^{3}$)  & $-0.277^{**\phantom{*}}$ & $-0.137^{\phantom{***}}$ & $0.001^{\phantom{***}}$ & $0.159^{\phantom{***}}$ & $-0.112^{\phantom{***}}$ & $-0.251^{*\phantom{**}}$\\
      & $(0.124)^{\phantom{**}}$ & $(0.149)^{\phantom{**}}$ & $(0.178)^{\phantom{**}}$ & $(0.184)^{\phantom{**}}$ & $(0.129)^{\phantom{**}}$ & $(0.133)^{\phantom{**}}$\\ [8pt]
     Female  & $-0.151^{***}$ & $-0.025^{\phantom{***}}$ & $-0.119^{*\phantom{**}}$ & $-0.103^{*\phantom{**}}$ & $-0.080^{*\phantom{**}}$ & $0.044^{\phantom{***}}$\\
      & $(0.044)^{\phantom{**}}$ & $(0.047)^{\phantom{**}}$ & $(0.063)^{\phantom{**}}$ & $(0.062)^{\phantom{**}}$ & $(0.043)^{\phantom{**}}$ & $(0.046)^{\phantom{**}}$\\ [8pt]
     Travel time to   & $0.141^{***}$ & $-0.206^{***}$ & $-0.119^{\phantom{***}}$ & $-0.226^{***}$ & $-0.273^{***}$ & $-0.252^{***}$\\
      nearest parent $>1$hr& $(0.051)^{\phantom{**}}$ & $(0.058)^{\phantom{**}}$ & $(0.080)^{\phantom{**}}$ & $(0.076)^{\phantom{**}}$ & $(0.056)^{\phantom{**}}$ & $(0.055)^{\phantom{**}}$\\ [8pt]
     Log(household income)  & $0.044^{***}$ & $0.007^{\phantom{***}}$ & $0.025^{\phantom{***}}$ & $0.017^{\phantom{***}}$ & $0.003^{\phantom{***}}$ & $0.017^{\phantom{***}}$\\
      & $(0.017)^{\phantom{**}}$ & $(0.019)^{\phantom{**}}$ & $(0.022)^{\phantom{**}}$ & $(0.016)^{\phantom{**}}$ & $(0.018)^{\phantom{**}}$ & $(0.020)^{\phantom{**}}$\\ [8pt]
      \hline
      \multicolumn{7}{p{.9\textwidth}}{
      \small{
      The posterior credible interval excludes
      zero at level 90\% (*), 95\% (**) or 99\% (***).}}\\
      \multicolumn{7}{p{.9\textwidth}}{\small $\dagger$ Age of respondent is centered at 40.}
    \end{tabular}
\end{table}

\begin{table}[!htbp]
  \centering
  \caption{
  Fitted residual correlations calculated using the parameter estimates
  in \cref{t_model_corr_coeff}, averaged over parameter values
  in the MCMC samples and over covariate values
  in the analysis sample. The `Overall' values are
averaged over sample values of all the covariates, and the other fitted
values over the sample values of all the covariates except for the one
fixed at the specified value.}
  \vspace*{1ex}
  \label{t_model_corr_pred}
  \small
      \begin{tabular}{lrrrr|rr}
      \hline
      Covariate & \multicolumn{6}{c}{Correlation}\\
      setting & \multicolumn{1}{c}{GP$\leftrightarrow$RP} & \multicolumn{1}{c}{GP$\leftrightarrow$RF} & \multicolumn{1}{c}{GF$\leftrightarrow$RP} & \multicolumn{1}{c}{GF$\leftrightarrow$RF} & \multicolumn{1}{c}{GP$\leftrightarrow$GF} & \multicolumn{1}{c}{RP$\leftrightarrow$RF}\\
      \hline
      \rule{0pt}{3ex}Overall& $0.38^{\phantom{***}}$ & $0.16^{\phantom{***}}$ & $0.02^{\phantom{***}}$ & $-0.06^{\phantom{***}}$ & $0.36^{\phantom{***}}$ & $0.20^{\phantom{***}}$ \\[.5ex]
      \emph{Age of respondent} &&&&&& \\
      \hspace*{1em}35 years&  $0.53^{\phantom{***}}$ & $0.14^{\phantom{***}}$ & $0.00^{\phantom{***}}$ & $-0.07^{\phantom{***}}$ & $0.39^{\phantom{***}}$ & $0.31^{\phantom{***}}$ \\
      \hspace*{1em}45 years&  $0.39^{\phantom{***}}$ & $0.18^{\phantom{***}}$ & $0.03^{\phantom{***}}$ & $-0.08^{\phantom{***}}$ & $0.37^{\phantom{***}}$ & $0.22^{\phantom{***}}$ \\
      \hspace*{1em}55 years&  $0.20^{\phantom{***}}$ & $0.19^{\phantom{***}}$ & $0.06^{\phantom{***}}$ & $-0.06^{\phantom{***}}$ & $0.32^{\phantom{***}}$ & $0.08^{\phantom{***}}$ \\[.5ex]
      \emph{Gender} &&&&&&\\
      \hspace*{1em}Female&  $0.31^{\phantom{***}}$ & $0.14^{\phantom{***}}$ & $-0.03^{\phantom{***}}$ & $-0.10^{\phantom{***}}$ & $0.32^{\phantom{***}}$ & $0.22^{\phantom{***}}$ \\
      \hspace*{1em}Male&  $0.47^{\phantom{***}}$ & $0.17^{\phantom{***}}$ & $0.09^{\phantom{***}}$ & $0.00^{\phantom{***}}$ & $0.40^{\phantom{***}}$ & $0.18^{\phantom{***}}$ \\[.5ex]
      \multicolumn{3}{l}{\emph{Travel time to the nearest parent}} &&&&\\
      \hspace*{1em}$>$ 1 hr& $0.48^{\phantom{***}}$ & $0.01^{\phantom{***}}$ & $-0.06^{\phantom{***}}$ & $-0.22^{\phantom{***}}$ & $0.16^{\phantom{***}}$ & $0.02^{\phantom{***}}$ \\
      \hspace*{1em}$\leq$ 1 hr& $0.34^{\phantom{***}}$ & $0.21^{\phantom{***}}$ & $0.05^{\phantom{***}}$ & $0.00^{\phantom{***}}$ & $0.43^{\phantom{***}}$ & $0.27^{\phantom{***}}$ \\[.5ex]
      \multicolumn{4}{l}{\emph{Logarithm of household equivalised income}} &&&\\
      \hspace*{1em}25th percentile& $0.37^{\phantom{***}}$ & $0.15^{\phantom{***}}$ & $0.02^{\phantom{***}}$ & $-0.06^{\phantom{***}}$ & $0.36^{\phantom{***}}$ & $0.20^{\phantom{***}}$ \\
      \hspace*{1em}50th percentile& $0.38^{\phantom{***}}$ & $0.16^{\phantom{***}}$ & $0.02^{\phantom{***}}$ & $-0.06^{\phantom{***}}$ & $0.36^{\phantom{***}}$ & $0.20^{\phantom{***}}$ \\
      \hspace*{1em}75th percentile& $0.39^{\phantom{***}}$ & $0.16^{\phantom{***}}$ & $0.03^{\phantom{***}}$ & $-0.05^{\phantom{***}}$ & $0.36^{\phantom{***}}$ & $0.21^{\phantom{***}}$ \\
      \hline
  \end{tabular}
\end{table}

\clearpage

\section{Conclusions}
\label{sec:conclusion}

We have proposed methods for analysing the levels and correlations of
intergenerational help and support. This involved defining a model for
the joint distribution of latent variables which represent individuals'
tendencies of giving and receiving different types of support. A
particular focus of the paper was on developing models for how the
correlations of these variables depend on covariates. A linear model was
specified for each correlation, and the estimation procedure was
designed so that it ensures that the estimated model implies positive
definite correlation matrices over the relevant range of the covariates.
This builds on literature on such `constrained' methods of estimation
for models for correlations, which are here extended to include
unit-level covariates. The estimation is carried out using a tailored
MCMC algorithm which includes an efficient Metropolis-Hastings
sub-procedure for estimating the correlation model.

The model was used to study exchanges of practical and financial support
between adult individuals and their non-coresident parents in the UK,
using survey data from the UK Household Longitudinal Study. The
modelling framework allows us to model both the conditional means and
correlations of different helping tendencies. The mean levels are
broadly positively associated with many characteristics of the
recipients that indicate higher need, and with characteristics of givers
that indicate their higher capacity to give help. These results are,
arguably, fairly encouraging about patterns of intergenerational support
in this population. Less positively, however, a very substantial
proportion of both adult individuals and their parents do not typically
give any of the kinds of help considered here. The estimated
correlations indicate  reciprocity, where those who tend to give high
levels of practical help also tend to receive much of it, and
complementarity, where those who tend to give high levels of one kind of
help (practical or financial) also tend to give much of the other kind.
This suggests a picture of a general culture of helpfulness within some
families, and general lack of it in others, rather than a sort of
zero-sum game where help would flow only in one direction at a time and
one kind of help would reduce the amount of other kinds.

This work could be extended in a number of ways in future research.
Methodologically, the proposed modelling approach for the correlation
matrix could be embedded into other covariance modelling tasks, such as
the copula model
\citep[]{hoffExtendingRankLikelihood2007,murrayBayesianGaussianCopula2013}.
The computational efficiency and mixing rates of the simple element-wise
Metropolis-Hastings MCMC sampler that was used here could perhaps be
improved by using other approaches, for example adaptive MCMC
\citep[]{haarioAdaptiveMetropolisAlgorithm2001,andrieuTutorialAdaptiveMCMC2008}
which proposes multiple parameters from an adaptive proposal in each
iteration.

Substantively, the choices of this analysis were constrained by the
available data. Although we were able to consider practical and
financial support separately, the single indicator of financial support
leaves us unable to examine varieties of it in more detail. Because the
data were collected from the adult children only, we have limited
information about their parents. Both of these limitations could be
relaxed by richer data, but collecting it would be correspondingly more
demanding. Another promising direction would be to extend these models
to longitudinal data. This would allow us, for example, to examine
questions of reciprocity and complementarity of help over time as well
as contemporaneously as was done here. These areas of further research
remain to be pursued.

\vspace{2ex}
\textbf{Acknowledgements}
This research was supported by a UK Economic and Social Research Council
(ESRC) grant ``Methods for the Analysis of Longitudinal Dyadic Data with
an Application to Inter-generational Exchanges of Family Support'' (ref.
ES/P000118/1). Additional funding for Siliang Zhang was provided by Shanghai Science and Technology Committee Rising-Star Program (22YF1411100).

\appendix

\setcounter{equation}{0}
\renewcommand{\theequation}{A\arabic{equation}}
\renewcommand{\thesection}{Appendix \Alph{section}}
\renewcommand{\thelemma}{A.\arabic{lemma}}

\section{Proofs of the propositions in Section \ref{sub:Some Useful Properties}}
\label{app:proofs_of_lemmas}

\begin{proof}[Proof of \cref{prop:C_mu1}]\leavevmode
  \begin{enumerate}[(i)]
    \item
    Let $\boldsymbol{\alpha}\in C_{\alpha,S_{X_1}}$, so that
    $\boldsymbol{\alpha}^{\top}\mathbf{X}\in C_{\rho}$ for all
    $\mathbf{X}\in S_{X_{1}}$. Since $S_{X_{2}}\subseteq S_{X_{1}}$, in
    particular, for all
    $\mathbf{X}\in S_{X_{2}}\subseteq S_{X_1}$, $\boldsymbol{\alpha}^{\top}\mathbf{X}\in C_{\rho}$, and thus $\boldsymbol{\alpha}\in
    C_{\alpha,S_{X_2}}$.
    \item
    Since $S_{X}\subseteq \text{Conv}(S_{X})$, we have
    $C_{\alpha,\text{Conv}(S_{X})}\subseteq C_{\alpha,S_{X}}$ by (i). So we just need to prove the other
    direction.
    Suppose that
    $\boldsymbol\alpha\in C_{\alpha,S_{X}}$, for any $\mathbf{X}'\in \text{Conv}(S_X)$, there exist a
    finite number of points $\mathbf{X}_1,...,\mathbf{X}_r\in S_X$ and
    $\lambda_1,\dots,\lambda_r\geq 0$, $\sum_j\lambda_j = 1$, such that
    $\mathbf{X}' = \sum_j\lambda_j \mathbf{X}_j$. We then have
    $\boldsymbol{\alpha}^{\top} \mathbf{X}' =
    \boldsymbol{\alpha}^{\top}(\sum_j\, \lambda_j \mathbf{X}_{j}) =
    \sum_j \lambda_j \,(\boldsymbol{\alpha}^{\top}\mathbf{X}_j)\in
    C_\rho$, i.e., $\boldsymbol{\alpha}\in C_{\alpha,\text{Conv}(S_X)}$, which holds because
    $\boldsymbol{\alpha}^{\top}\mathbf{X}_j\in C_\rho$ for all
    $j=1,\dots,r$, and
    $C_\rho$ is a convex set.
    \item
    $\boldsymbol{\alpha}=\mathbf{0}$ gives $\boldsymbol{\rho}=\mathbf{0}$. This
    implies the identity correlation matrix, which is in $C_{\rho}$.
    \item
    Under the further assumption stated in (iv), we can find a set
    $S_{X_*}=\{\mathbf{X}_{1},\dots,\mathbf{X}_{q}\}\subseteq
    S_{X}$ such that the matrix $\mathbf{X}_{*}=[
    \mathbf{X}_{1}, \dots,\mathbf{X}_{q}]$ is non-singular. Suppose that
    $\boldsymbol{\alpha}\in C_{\alpha,S_{X_*}}$, and let
    $\boldsymbol{\alpha}^{\top}\mathbf{X}_{*}=
    [\boldsymbol{\rho}_{1}, \dots, \boldsymbol{\rho}_{q}]$.
    Then
    $\boldsymbol{\alpha}^{\top}=[\boldsymbol{\rho}_{1}, \dots, \boldsymbol{\rho}_{q}]\mathbf{X}_{*}^{-1}$. This is bounded, because all
    elements of $\boldsymbol{\rho}_{1}, \dots, \boldsymbol{\rho}_{q}$
    are bounded (moreover, $\boldsymbol{\rho}\in [-1,1]^L$). Finally, since $S_{X_*}\subseteq S_{X}$, we have
$C_{\alpha,S_{X}}\subseteq C_{\alpha,S_{X_*}}$ by (ii), and thus $C_{\alpha,S_{X}}$ is also
bounded.
    \item Suppose that $\boldsymbol{\alpha}_{1}, \,
    \boldsymbol{\alpha}_{2} \in
    C_{\alpha,S_X}$ and that $0\leq \lambda\leq 1$. Then
      $
      (\lambda\,
      \boldsymbol{\alpha}_{1}+(1-\lambda)\,\boldsymbol{\alpha}_{2})^{\top}\mathbf{X}
 = \lambda\,
 \boldsymbol{\alpha}_{1}^{\top}\mathbf{X}+(1-\lambda)\,\boldsymbol{\alpha}_{2}^{\top}
 \mathbf{X} =\lambda\, \boldsymbol{\rho}_{1} +
 (1-\lambda)\,\boldsymbol{\rho}_{2}
      \in C_\rho,
$
    where the last equation holds since $C_\rho$ is a convex set. Thus
    $\lambda\, \boldsymbol{\alpha}_{1}+(1-\lambda)\,\boldsymbol{\alpha}_{2}\in
    C_{\alpha,S_{X}}$.
  \end{enumerate}
\end{proof}

The proof of Proposition \ref{prop:conti_interval_mu} builds on the key
ideas of \cite{johnbar2000}, extended to the case of models with
covariates that we consider.

\begin{lemma}\label{lemma:conti_rho}
Let
$\mathbf{R}(\boldsymbol{\rho})=\mathbf{R}(\rho_{l},\boldsymbol\rho_{-l})$
be the positive definite
correlation matrix defined by distinct correlations
$\boldsymbol{\rho}=(\rho_{l},\boldsymbol{\rho}_{-l}^{\top})^{\top}$. Consider
$f_{l}(\rho_l') = \vert \mathbf{R}(\rho_l', \boldsymbol\rho_{-l})\vert$
as a univariate function of $\rho_{l}'\in [-1,1]$. Then
  $f_{l}(\rho_l')$ is a quadratic function of $\rho_l'$ with negative second
  order coefficient. The matrix $\mathbf{R}_{l}=\mathbf{R}(\rho_l', \boldsymbol\rho_{-l})$ is
  positive definite if and only if $f_{l}(\rho_l')>0$.
\end{lemma}

\begin{proof}[Proof of Lemma A.1]
  $\mathbf{R}_{l}$ is a symmetric matrix where $\rho_{l}'$ appears once in
  both its upper and lower triangles, so $f_{l}(\rho_l')$ is a quadratic
  function. Suppose that $\mathbf{R}_{l}$ is a $K\times K$ matrix.
Without loss of
  generality, assume that $\rho_l'$ is in its $K$th row, first column (and
  first row, $K$th column), as we can always swap both row and column
  without changing the positive definiteness and determinant value.
  Thus, the coefficient of $(\rho_{l}')^{2}$ in $f_{l}(\rho_l')$ is
  $c_{l}=(-1)^{2K+1}|\mathbf{R}_{(l)}|$, where $\mathbf{R}_{(l)}$ is the
  submatrix of $\mathbf{R}_{l}$ obtained by deleting the first and last rows
  and columns. Here $\mathbf{R}_{(l)}$ is a correlation matrix,
  obtained by deleting from $\mathbf{R}(\boldsymbol{\rho})$ all those
  correlations which involve either of the two variables whose
  correlation is $\rho_{l}$. Thus $\mathbf{R}_{(l)}$ is positive definite,
  $|\mathbf{R}_{(l)}|>0$, and $c_{l}<0$.

  $\mathbf{R}_{l}$ is positive
  definite if and only if $\vert \mathbf{R}_{lk}\vert>0$ for all
  $k=1,\ldots,K,$ where $\mathbf{R}_{lk}$ is the $k$th primary submatrix of
  $\mathbf{R}_{l}$ (Sylvester's criterion).
  Here $\rho_l'$ only affects
  $\vert \mathbf{R}_{lK}\vert=
  \vert \mathbf{R}_{l}\vert$.
  Because
  $\mathbf{R}_{l1},\dots, \mathbf{R}_{l,K-1}$ are
  equal to the corresponding submatrices of the positive definite
  correlation matrix
  $\mathbf{R}(\boldsymbol{\rho})$, we have
$\vert \mathbf{R}_{lk}\vert>0, \text{ for } k=1,\ldots,K-1$.
  So
  $\mathbf{R}_{l}=\mathbf{R}(\rho_l', \boldsymbol{\rho}_{-l})$ is positive definite if
  and only if $f_{l}(\rho_{l}')=\vert \mathbf{R}_{l}\vert>0$.
\end{proof}

\begin{proof}[Proof of \cref{prop:conti_interval_mu}]\leavevmode

From \cref{lemma:conti_rho} we know that
$\mathbf{R}_{jl}=\mathbf{R}(\rho_{l}',\boldsymbol\rho_{-l}^{(j)})$ is positive definite
if and only if $f_{jl}(\rho_{l}')=\vert
\mathbf{R}_{jl}\vert>0$. We can write
$f_{jl}(\rho_{l}') = c_{jl}(\rho_{l}')^2 + d_{jl}\rho_{l}' + e_{jl}$,
where $c_{jl} = [f_{jl}(1)+f_{jl}(-1)-2f_{jl}(0)]/2$, $d_{jl} = [f_{jl}(1) -
f_{jl}(-1)]/2$ and $e_{jl} = f_{jl}(0)$. The set of values for $\rho_{l}'$
for which $f_{jl}(\rho_{l}')>0$ is a finite interval because $c_{jl}<0$,
$f_{jl}(0)=|\mathbf{R}(0,\boldsymbol{\rho}_{-l})|>0$, and $f_{jl}(\rho_{l}')$
is a continuous function. Let us denote the roots of
$f_{jl}(\rho_{l}')=0$ by $x_{jl1}>x_{jl2}$, and define
\begin{equation}
\begin{aligned}
g_{jl} &= \frac{x_{jl1}+x_{jl2}}{2} =
-\frac{d_{jl}}{2c_{jl}},\\
h_{jl} &= \frac{x_{jl1}-x_{jl2}}{2} =
\sqrt{\frac{d_{jl}^2-4c_{jl}e_{jl}}{4c_{jl}^2}}.
\end{aligned}
\end{equation}
$\mathbf{R}_{jl}$ is positive definite when $\rho_{l}'\in (g_{jl} -
h_{jl},g_{jl}+h_{jl})$.

Consider now
$\boldsymbol{\rho}_{j}=\boldsymbol{\alpha}^{\top}\mathbf{X}_{j}$ as
specified by model (\ref{eta_corr_model}), as functions of coefficients
$\boldsymbol{\alpha}$ and covariates $\mathbf{X}_{j}$.
Consider
$\rho_{l}'=\alpha_{lm}'X_{jm} + \sum_{k\ne m} \alpha_{lk}X_{jk}$ as
implied by this model, treating $\alpha_{lm}'$ for a single
$m=1,\dots,q$ as the argument of the function and fixing all the other
elements of $\boldsymbol{\alpha}$ and $\mathbf{X}_{j}$ at the values
which defined $\boldsymbol{\rho}_{j}$.
Solving the end points of the feasible interval
of $\rho_{l}'$ for $\alpha_{lm}'$, and taking into account the
sign of $X_{jm}$ gives the feasible interval for
$\alpha_{lm}'$ with end points $a_{lm}^{(j)}$ and $b_{lm}^{(j)}$ as shown
in (\ref{eq:interval_j}) in \cref{prop:conti_interval_mu}, when
$X_{jm}\ne 0$. When $X_{jm}=0$,
$f_{jl}(\rho_{l}')$ does not depend on $\alpha_{lm}'$ and the
interval can be taken to be infinite. The interval for $\alpha_{lm}$ which is feasible for
all of the $\mathbf{X}_{j}\in S_{XT}$  is then $(a_{lm},b_{lm})=\cap_{j}
(a_{lm}^{(j)},b_{lm}^{(j)})$.

\end{proof}

\section{Details of the MCMC algorithm}
\label{app:mcmc}

\setcounter{equation}{0}
\renewcommand{\theequation}{B\arabic{equation}}

Here we describe the MCMC sampling algorithm for estimating the
structural-model parameters of the model which was introduced in Section
\ref{sec:latent_variable_model}. The general idea of this estimation was
outlined in Section \ref{sub:estimation_of_the_structural_model2}, where
we also described the sampling steps for the parameters
$\boldsymbol{\alpha}$ of the model for the correlations of
$\boldsymbol{\eta}_{i}$. As discussed there, the steps for the other
elements of the model are the same or very similar to the ones proposed
in \cite{kuhaetal22}. Their details are also given
here in order to keep this description self-contained.

The algorithm has been packed into an R \citep{rcoreteam} package [which
will be included in the supplementary materials and made available open
source on an author's GitHub page]. The algorithm was programmed in R
with core functions implemented in
{C\nolinebreak[4]\hspace{-.05em}\raisebox{.4ex}{\tiny\bf ++}}, where two
techniques are used to speed up the procedure. First, for sampling steps
with non-standard distributions, adaptive rejection sampling
\citep{gilksAdaptiveRejectionMetropolis1995} is used, exploiting
log-concavity of the posterior density functions. Second, parallel
sampling is used within each MCMC iteration where possible. The
parallelization is implemented through the OpenMP
{C\nolinebreak[4]\hspace{-.05em}\raisebox{.4ex}{\tiny\bf ++}} API
\citep{dagum1998openmp}.

Different elements of $\boldsymbol{\zeta}$ and $\boldsymbol{\psi}$ are
sampled one at a time, as scalars or vectors as appropriate. In the
notation below, those quantities that are not being sampled in a given
step are taken to be observed and fixed at their most recently sampled
values.

\vspace{1ex}
\textbf{Sampling the latent variables}: Generate values for the latent
variables
$\boldsymbol{\zeta}_{i}=(\boldsymbol{\xi}_{i}^\top,\boldsymbol{\eta}_{i}^\top)^\top$,
given the observed data and current values of the parameters
$\boldsymbol{\psi}$. This can be parallelised, because
$\boldsymbol{\zeta}_{i}$ for different  units $i$ are conditionally
independent.

(1) Sampling $\boldsymbol{\xi}$ from
$p(\boldsymbol{\xi}|\boldsymbol{\eta},\mathbf{Y},\mathbf{X},\boldsymbol{\psi})$:
Draw  $\boldsymbol{\xi}_{i}= (\xi_{Gi},\xi_{Ri})^\top$ independently for
$i=1,\dots,n$, from multinomial distributions with probabilities
\begin{eqnarray}
\lefteqn{p(\xi_{G}=j,\xi_{R}=k|\boldsymbol{\eta},\mathbf{Y}_{i},\mathbf{X}_{i},\boldsymbol{\psi})}
\label{p_xi}
\\
&\propto&
p(\mathbf{Y}_{Gi}|\xi_{G}=j,\eta_{Gi})\,
p(\mathbf{Y}_{Ri}|\xi_{R}=k,\eta_{Ri})\,
p(\xi_{G}=j,\xi_{R}=k|\mathbf{X}_{i};\boldsymbol{\psi}_{\xi})
\nonumber
\end{eqnarray}
for $j,k=0,1$, where the measurement model is specified by
(\ref{eq:meas1})--(\ref{eq:probit_linkF}) for $\mathbf{Y}_{Gi}$
and similarly for $\mathbf{Y}_{Ri}$, and the structural model for
$\boldsymbol{\xi}_{i}$ is specified by (\ref{eq:xi_dist}).

(2) Sampling $\boldsymbol{\eta}$ from
$p(\boldsymbol{\eta}|\boldsymbol{\xi},\mathbf{Y},\mathbf{X},\boldsymbol{\psi})$:
Draw  $\boldsymbol{\eta}_{i}=
(\eta_{GPi}, \eta_{RPi}, \eta_{GFi}, \eta_{RFi})^{\top}$ independently
for $i=1,\dots,n$, from
\begin{eqnarray}
p(\eta_{GP}|\boldsymbol\eta_{-GPi},\boldsymbol{\xi}_{i},\mathbf{Y}_{i},\mathbf{X}_{i},\boldsymbol{\psi})
&\propto&
p(\mathbf{Y}_{GPi}|\xi_{Gi},\eta_{GP})\,
p(\eta_{GP}|\boldsymbol\eta_{-GPi},\mathbf{X}_{i};\boldsymbol{\psi}_{\eta})
\label{p_etaGP} \\
p(\eta_{GF}|\boldsymbol\eta_{-GFi},\boldsymbol{\xi}_{i},\mathbf{Y}_{i},\mathbf{X}_{i},\boldsymbol{\psi})
&\propto& p(Y_{GFi}|\xi_{Gi},\eta_{GF})\,
p(\eta_{GF}|\boldsymbol\eta_{-GFi},\mathbf{X}_{i};\boldsymbol{\psi}_{\eta})
\label{p_etaGF} \\
p(\eta_{RP}|\boldsymbol\eta_{-RPi},\boldsymbol{\xi}_{i},\mathbf{Y}_{i},\mathbf{X}_{i},\boldsymbol{\psi})
&\propto& p(\mathbf{Y}_{RPi}|\xi_{Ri},\eta_{RP})\,
p(\eta_{RP}|\boldsymbol\eta_{-RPi},\mathbf{X}_{i};\boldsymbol{\psi}_{\eta})
\label{p_etaRP} \\
p(\eta_{RF}|\boldsymbol\eta_{-RFi},\boldsymbol{\xi}_{i},\mathbf{Y}_{i},\mathbf{X}_{i},\boldsymbol{\psi})
&\propto& p(Y_{RFi}|\xi_{Ri},\eta_{RF})\,
p(\eta_{RF}|\boldsymbol\eta_{-RFi},\mathbf{X}_{i};\boldsymbol{\psi}_{\eta}).
\label{p_etaRF}
\end{eqnarray}
Here $\boldsymbol{\eta}_{-GPi}$ denotes
$(\eta_{GFi},\eta_{RPi},\eta_{RFi})$ and $\boldsymbol{\eta}_{-GFi}$,
$\boldsymbol{\eta}_{-RPi}$ and $\boldsymbol{\eta}_{-RFi}$ are defined
similarly. The conditional distributions for the $\eta$-variables on the
right hand sides of (\ref{p_etaGP})--(\ref{p_etaRF}) are the univariate
conditional normal distributions implied by the joint normal
distribution given by (\ref{eta_mean})--(\ref{eta_corr}). The sampling
distributions depend on the values of the $\xi$-variables. When
$\xi_{Gi}=0$, in which case always $\mathbf{Y}_{Gi}=\mathbf{0}$, we have
$p(\mathbf{Y}_{GPi}|\xi_{Gi},\eta_{GP})=
p(Y_{GFi}|\xi_{Gi},\eta_{GF})=1$ and $\eta_{GPi}$ and
$\eta_{GFi}$ are drawn directly from the conditional normal
distributions. When
$\xi_{Gi}=1$,
adaptive rejection
sampling is used for $\eta_{GPi}$ and
truncated normal sampling for $\eta_{GFi}$. The sampling of
$\eta_{RPi}$ and $\eta_{RFi}$ is analogous.

\vspace{1ex}
\textbf{Sampling the parameters of the structural model}: Generate
values for the parameters $\boldsymbol{\psi}$ from their distributions
given the observed variables and current imputed values of the latent
variables $\boldsymbol{\zeta}$. These have the form of posterior
distributions of these structural parameters when both
$\boldsymbol{\zeta}$ and $\mathbf{X}$ are taken to be observed data
(this step does not depend on $\mathbf{Y}$). The prior distributions are
taken to be of the form $p(\boldsymbol{\psi})=
p(\boldsymbol{\psi}_{\xi})p(\boldsymbol{\beta})p(\boldsymbol{\sigma})p(\boldsymbol{\alpha})$,
i.e.\ independent for different blocks of parameters; their specific
forms are given below. The sampling steps for $\boldsymbol{\psi}_{\xi}$
and $\boldsymbol{\psi}_{\eta}$ do not depend on each other, so they can
be carried out in either order or in parallel.

(3) Sampling $\boldsymbol{\psi}_{\xi}=(\boldsymbol{\gamma}_{01}^{\top},
\boldsymbol{\gamma}_{10}^{\top},
\boldsymbol{\gamma}_{11}^{\top})^{\top}$ from
$p(\boldsymbol{\psi}_{\xi}|\mathbf{X},\boldsymbol{\xi}) \propto
p(\boldsymbol{\xi}|\mathbf{X};\boldsymbol{\psi}_{\xi})\,
p(\boldsymbol{\psi}_{\xi})$. This is the posterior distribution of the
coefficients of the multinomial logistic model
(\ref{eq:xi_dist}) for $\boldsymbol{\xi}_{i}$
given $\mathbf{X}_{i}$.
Define $\boldsymbol{\gamma}=(
\boldsymbol{\gamma}_{00}^{\top}, \boldsymbol{\gamma}_{01}^{\top},
\boldsymbol{\gamma}_{10}^{\top},
\boldsymbol{\gamma}_{11}^{\top})^{\top}$, where
$\boldsymbol{\gamma}_{00}=\mathbf{0}$. Let
$\gamma_{jkr}$ demote the coefficient of
$X_{jkr}$ in the model for
$p(\xi_{\chi}=j,\xi_{Ri}=k|\mathbf{X}_{i};\boldsymbol{\psi}_{\xi})$, and
$\boldsymbol{\gamma}_{-jkr}$ denote the vector obtained by omitting
$\gamma_{jkr}$ from $\boldsymbol{\gamma}$.
We take the prior distributions of each non-zero
$\gamma_{jkr}$ to be independent of each other, with
$p(\gamma_{jkr})\sim N(0,\sigma^{2}_{\gamma})$ with
$\sigma^{2}_{\gamma}=100$.
The sampling is done using conditional Gibbs
sampling, one parameter at a time.
We cycle over all $r=1\dots,Q$ and over $(j,k)=(0,1), \; (1,0),\; (1,1)$
to draw $\gamma_{jkr}$ from
\begin{equation}
p(\gamma_{jkr}|\boldsymbol{\gamma}_{-jkr},\mathbf{X},\boldsymbol{\xi})
\propto \left[
\prod_{i=1}^{n} \;
\frac{\prod_{r,s=0,1}\;\exp(\boldsymbol\gamma_{rs}^{\top}\mathbf{X}_{i})^{\delta_{ijk}}}{
\sum_{r,s=0,1}\; \exp(\boldsymbol{\gamma}_{rs}^{\top}\mathbf{X}_{i})
}
\right]\; p(\gamma_{jkr})
\label{p_psiksi}
\end{equation}
where $\delta_{ijk}=\mathbbm{1}(\xi_{Gi}=j, \xi_{Ri}=k)$. These are sampled using
adaptive rejection sampling.

(4) Sampling
$\boldsymbol{\psi}_{\eta}=(\text{vec}(\boldsymbol{\beta})^\top,
\boldsymbol{\sigma}^{\top},
\text{vec}(\boldsymbol{\alpha})^{\top})^{\top}$ from
$p(\boldsymbol{\psi}_{\eta}|\mathbf{X},\boldsymbol{\eta}) \propto
p(\boldsymbol{\eta}|\mathbf{X};\boldsymbol{\psi}_{\eta})\,
p(\boldsymbol{\psi}_{\eta})$. Here the sampling of $\boldsymbol{\alpha}$
has been described in Section
\ref{sub:estimation_of_the_structural_model2}.
For $\boldsymbol{\beta}$, the sampling is from the
posterior distribution
$p(\text{vec}(\boldsymbol{\beta})|\mathbf{X},\boldsymbol{\eta}) \propto
p(\boldsymbol{\eta}|\mathbf{X}; \boldsymbol{\psi}_{\eta})\,
p(\text{vec}(\boldsymbol{\beta}))$ where $\boldsymbol{\sigma}$ and
$\boldsymbol{\alpha}$ are regarded as known. This means that the
conditional covariance matrices
$\boldsymbol{\Sigma}_{i}=\text{cov}(\boldsymbol{\eta}_{i}|\mathbf{X}_{i};
\boldsymbol{\sigma},\boldsymbol{\alpha})$
are also known here.
We specify $p(\text{vec}(\boldsymbol{\beta}))\sim N(\mathbf{0},
\sigma^{2}_{\beta}\,\mathbf{I}_{4Q})$ with $\sigma^{2}_{\beta}=100$. The
sampling is done separately for each of the
four subvectors of $\boldsymbol{\beta}$. Let
$\boldsymbol{\beta}_{1}$ denote one of them, say
$\boldsymbol{\beta}_{1}=\boldsymbol{\beta}_{GP}$, and
$\boldsymbol{\beta}_{2}$ the rest of them, say
$\boldsymbol{\beta}_{2}=[\boldsymbol{\beta}_{RP},\,
\boldsymbol{\beta}_{GF},\, \boldsymbol{\beta}_{RF}]$, and let
$\boldsymbol{\psi}_{\eta(\beta_{1})}$ denote all the elements of
$\boldsymbol{\psi}_{\eta}$ other than
$\boldsymbol{\beta}_{1}$.
Let $\boldsymbol{\eta}_{i}$ be partitioned correspondingly into
$\eta_{1i}$ and $\boldsymbol{\eta}_{2i}$, and $\boldsymbol{\Sigma}_{i}$
into the blocks $\Sigma_{11i}$, $\boldsymbol{\Sigma}_{12i}$ and
$\boldsymbol{\Sigma}_{22i}$. The conditional distribution
$p(\eta_{1i}|\boldsymbol{\eta}_{2i}, \mathbf{X}_{i};
\boldsymbol{\psi}_{\eta})$ is then univariate normal with
mean
$\boldsymbol{\beta}_{1}^{\top}\mathbf{X}_{i}+d_{2i}$, where
$d_{2i}=\boldsymbol{\Sigma}_{12i}\boldsymbol{\Sigma}_{22i}^{-1}(\boldsymbol{\eta}_{2i}
-\boldsymbol{\beta}_{2}^{\top}\mathbf{X}_{i})$,
and variance
$\sigma^{2}_{1i}=\Sigma_{11i}-\boldsymbol{\Sigma}_{12i}\boldsymbol{\Sigma}_{22i}^{-1}
\boldsymbol{\Sigma}_{12i}^{\top}$.
Let $\mathbf{V}_{1}=\text{diag}(\sigma^{2}_{1i},\dots,\sigma^{2}_{1n})$
and
$\mathbf{e}_{1}=(\eta_{11}-d_{21},\dots,\eta_{1n}-d_{2n})^{\top}$.
The value of $\boldsymbol{\beta}_{1}$ is then sampled from
$
p(\boldsymbol{\beta}_{1}|\mathbf{X},\boldsymbol{\eta},\boldsymbol{\psi}_{\eta(\beta_{1})})
\sim
N(\mathbf{V}_{\beta_{1}}(\mathbf{X}^{\top}\mathbf{V}_{1}^{-1}\mathbf{e}_{1}),
\, \mathbf{V}_{\beta_{1}})
$
where
$\mathbf{V}_{\beta_{1}}=(\mathbf{X}^{\top}\mathbf{V}_{1}^{-1}\mathbf{X}+\mathbf{I}_{Q}/\sigma^{2}_{\beta})^{-1}$.
This is repeated with each of the four subvectors of $\boldsymbol{\beta}$
in turn in the role of $\boldsymbol{\beta}_{1}$.

For sampling of the standard deviation parameters
$\boldsymbol{\sigma}$, denote here
$\sigma_{1}=\sigma_{GP}$ and $\sigma_{2}=\sigma_{RP}$. For both of them
we use the prior distribution $\text{Inv-Gamma}(\alpha_{0},\beta_{0})$
with $\alpha_{0}=\beta_{0} = 10^{-5}$,
independently for $\sigma_{1}^{2}$ and $\sigma_{2}^{2}$. This implies
the priors $p(\sigma_{k}) \,\propto\,
\sigma_{k}^{-2\alpha_{0}-1} \exp(\beta_{0}/\sigma_{k}^{2})$ for $k=1,2$.
Denote by
$\boldsymbol{\psi}_{\eta(\sigma)}$ all other parameters in
$\boldsymbol{\psi}_{\eta}$ apart from $\sigma_{k}$.
Recall that this means that in $\boldsymbol{\Sigma}_{i}
=\mathbf{S}\,\mathbf{R}_{i}\,\mathbf{S}$, where
$\mathbf{S}=\text{diag}(\sigma_{1},\sigma_{2},1,1)$,
the correlation matrix
$\mathbf{R}_{i}=\mathbf{R}(\mathbf{X}_{i};\boldsymbol{\alpha})$ is also
treated as known here. Let
$\mathbf{e}_{i}=
(e_{i1},e_{i2},e_{i3},e_{i4})^{\top}=
\boldsymbol\eta_{i} -  \boldsymbol{\beta}^\top\mathbf{X}_i
$.
The parameter $\sigma_{k}$ is then drawn from
\begin{equation}
  \begin{aligned}
    p(\sigma_{k}|\mathbf{X}, \boldsymbol{\eta},
    \boldsymbol\psi_{\eta(\sigma)})
    &\;\propto\;    \prod_{i=1}^n
    \,p(\boldsymbol{\eta}_i|\mathbf{X}_{i}; \boldsymbol{\psi}_{\eta})
    p(\sigma_{k})
    \;\propto\; \prod_{i=1}^n \, \sigma_{k}^{-1}\exp\left(
    -\frac{1}{2}\mathbf{e}_i^{\top}\,\boldsymbol{\Sigma}_{i}^{-1}\,\mathbf{e}_{i}\right)
   p(\sigma_{k})\\
    &\; \propto \;  \sigma_{k}^{-\alpha-1}\exp\left(
    -\beta_1/\sigma_{k}^2 - 2\beta_{2}/\sigma_{k} \right),
  \end{aligned}
\end{equation}
where $\alpha = n + 2\alpha_{0},$ $\beta_1 =
\beta_{0}+(\sum_{i=1}^n e_{ik}^2 w_{kki})/2$,
 $\beta_2 = \sum_{i=1}^n\,e_{ik}(\sum_{j\neq
 k} w_{kji} e_{ij}/\sigma_j)/2$, and $w_{kji}$ is the $(k,j)$th
 element of $\mathbf{R}_i^{-1}$. Then random-walk Metropolis sampler or the
 adaptive rejection Metropolis sampler
 \citep[ARMS,][]{gilksAdaptiveRejectionMetropolis1995} can be used to sample $\sigma_{1}$ and $\sigma_{2}$.

\appendix
\section*{Appendix}
\begin{table}[!htbp]
  \centering
  \caption{Estimated parameters (measurement loadings and intercepts) of the measurement models for survey items on help given
  by respondents to their parents and on help received from the parents.}
  \vspace*{1ex}
  \begin{tabular}{lrrrr}
  \hline
   & \multicolumn{2}{l}{Giving} & \multicolumn{2}{l}{Receiving} \\
  & \multicolumn{2}{l}{practical help} & \multicolumn{2}{l}{practical help} \\
   Item& loading & intercept & loading & intercept\\
  \hline
Lifts in car                          & 1.12 & 0.83  & 1.14 & 1.54\\
  Shopping                           & 2.38 & 1.02  & 1.70 & 2.08\\
  Providing or cooking meals           & 1.24 & -0.28 & 1.15 & 1.57\\
  Basic personal needs (to parent only)         & 1.32 & -1.32 & --   & --  \\
  Looking after children (from parents only)           & --   & --    & 0.89 & 2.25\\
  Washing, ironing or cleaning     & 1.32 & -0.77 & 1.15 & 0.82\\
  Personal affairs          & 1.00 & 0.00  & 1.00 & 0.00\\
  Decorating, gardening or house repairs & 0.57 & -0.22 & 0.74 & 0.37\\[.5ex]
  \hline
  \label{t_items1}
  \end{tabular}
  \end{table}

\begin{table}
  \centering
  \caption{Estimated coefficients of the multinomial logistic model for the zero-inflation latent  classes $(\xi_{G},\xi_{R})$. The coefficients $\boldsymbol\gamma_{00}$ are fixed at $\boldsymbol{0}$ for
  identification.
  The estimates are posterior means from MCMC samples (with posterior standard
deviations in parentheses).
  }
  \label{t_model_estimates2}

  \vspace*{1ex}
  \begin{small}
  \begin{tabular}{lrrrrrcrrrrrr}\hline
  &
  \multicolumn{3}{c}{$\gamma_{jk}(\xi_{G}=j,\xi_{R}=k)$}  \\
  Covariate &$\gamma_{01}$\hspace*{4em}&$\gamma_{10}$\hspace*{4em}
  &$\gamma_{11}$\hspace*{4em}\\
  \hline
  Intercept &  $-3.98^{***}$ $(0.76)$ & $-1.73^{*\phantom{**}}$ $(1.08)$ & $1.50^{***}$ $(0.52)$ &  \\[.5ex]
  \multicolumn{2}{l}{\textbf{Respondent (child) characteristics}} & &\\
  Age (centered at 40) ($\times 10$) & $-0.13^{\phantom{***}}$ $(0.21)$ & $-0.44^{*\phantom{**}}$ $(0.24)$ & $-0.44^{***}$ $(0.09)$ &  \\[.5ex]
  Age squared ($\times 10^3$) & $-2.63^{\phantom{***}}$ $(1.72)$ & $1.34^{\phantom{***}}$ $(0.93)$ & $1.54^{***}$ $(0.46)$ &  \\[.5ex]
  Gender    &                &&                &           & \\
  ~~Female (vs.\ Male) & $1.47^{***}$ $(0.30)$ & $0.21^{\phantom{***}}$ $(0.17)$ & $0.06^{\phantom{***}}$ $(0.09)$ &  \\[.5ex]
  Partnership status              &               &   &           & \\
  ~~Partnered (vs.\ Single) & $0.13^{\phantom{***}}$ $(0.22)$ & $0.71^{***}$ $(0.20)$ & $-0.05^{\phantom{***}}$ $(0.10)$ &  \\[.5ex]
  \multicolumn{6}{l}{Age of youngest coresident child (vs.\ No children):}\\
   ~~0--1 years    & $0.51^{\phantom{***}}$ $(0.35)$ & $-0.60^{\phantom{***}}$ $(0.53)$ & $-0.01^{\phantom{***}}$ $(0.17)$ &  \\
  ~~2--4 years    & $0.50^{\phantom{***}}$ $(0.32)$ & $0.10^{\phantom{***}}$ $(0.40)$ & $0.45^{***}$ $(0.16)$ &  \\
  ~~5--10 years   & $0.43^{*\phantom{**}}$ $(0.26)$ & $-1.69^{***}$ $(0.66)$ & $0.16^{\phantom{***}}$ $(0.14)$ &  \\
  ~~11--16 years  & $-0.59^{*\phantom{**}}$ $(0.30)$ & $-0.03^{\phantom{***}}$ $(0.22)$ & $-0.32^{**\phantom{*}}$ $(0.14)$ &  \\
  ~~17-- years   & $-0.24^{\phantom{***}}$ $(0.39)$ & $-0.20^{\phantom{***}}$ $(0.24)$ & $-0.11^{\phantom{***}}$ $(0.17)$ &  \\[.5ex]
  Number of siblings (vs.\ None) & &  \\
  ~~1 & $-0.11^{\phantom{***}}$ $(0.28)$ & $-0.49^{**\phantom{*}}$ $(0.23)$ & $0.02^{\phantom{***}}$ $(0.15)$ &  \\
  ~~2     & $-0.70^{**\phantom{*}}$ $(0.27)$ & $-0.41^{*\phantom{**}}$ $(0.22)$ & $-0.23^{\phantom{***}}$ $(0.15)$ &  \\[.5ex]
  Long standing illness (vs.\ No) & $0.00^{\phantom{***}}$ $(0.23)$ & $-0.37^{*\phantom{**}}$ $(0.21)$ & $-0.27^{**\phantom{*}}$ $(0.11)$ &  \\[.5ex]
  \multicolumn{2}{l}{Employment status (vs.\ Employed)} &           &           & \\
  ~~Not employed & $-0.25^{\phantom{***}}$ $(0.21)$ & $0.36^{*\phantom{**}}$ $(0.19)$ & $-0.34^{***}$ $(0.10)$ &  \\[.5ex]
  Education
  (vs.\ Secondary or less)
 & &           &           & \\
  ~~Post-secondary     & $0.63^{***}$ $(0.18)$ & $-0.05^{\phantom{***}}$ $(0.16)$ & $0.20^{**\phantom{*}}$ $(0.09)$ &  \\[.5ex]
  Household tenure (vs.\ Renter) &&           &           & \\
  ~~Own home outright or by mortgage & $-0.23^{\phantom{***}}$ $(0.21)$ & $0.36^{*\phantom{**}}$ $(0.22)$ & $0.06^{\phantom{***}}$ $(0.10)$ &  \\[.5ex]
  Logarithm of household equivalised income& $0.37^{***}$ $(0.09)$ & $0.01^{\phantom{***}}$ $(0.11)$ & $-0.01^{\phantom{***}}$ $(0.05)$ &  \\[8pt]
  \multicolumn{3}{l}{\textbf{Parent characteristics}}         &        &\\
  Age of the oldest living parent &&&&&\\
  \hspace{1em}(centered at 70) ($\times 10$)  & $0.37^{*\phantom{**}}$ $(0.20)$ & $1.58^{***}$ $(0.42)$ & $0.17^{**\phantom{*}}$ $(0.07)$ &  \\[.5ex]
  Squared Age of the oldest parent ($\times 10^{3}$)  & $-12.72^{***}$ $(2.41)$ & $-2.91^{**\phantom{*}}$ $(1.50)$ & $-0.37^{\phantom{***}}$ $(0.37)$ &  \\[.5ex]
  At least one parent lives alone (vs.\ No) & $-0.84^{***}$ $(0.21)$ & $1.13^{***}$ $(0.17)$ & $0.25^{***}$ $(0.09)$ &  \\[8pt]
  \multicolumn{3}{l}{\textbf{Child-parent characteristics}} &&& \\
  Travel time to the nearest parent &&&&\\
  ~~More than 1 hour (vs.\ 1 hour or less)  & $-1.65^{***}$ $(0.24)$ & $-1.73^{***}$ $(0.20)$ & $-1.99^{***}$ $(0.10)$ &  \\
\hline
\multicolumn{6}{l}{{\small{The posterior credible interval excludes zero at
level 90\% (*), 95\% (**) or 99\% (***).}}}
\end{tabular}
\end{small}
\end{table}

\begin{table}
\centering
\caption{Fitted membership probabilities of the zero-inflation latent
classes $(\xi_{G},\xi_{R})$, from the estimated model
in Table S2. The fitted probabilities are averaged over parameter values in MCMC samples and over covariate values in the observed sample (for all covariates for the ``Overall'' figures, and for all but the fixed covariate for the rest. The odds ratios (OR) calculated from these averages are also shown.}
\label{t_model_xi_pred}

\vspace*{1ex}
\begin{small}
\begin{tabular}{lrrrrrrrrrrrr}\hline
&&&&&&&\multicolumn{6}{c}{Marginal probabilities}\\
Covariate &
\multicolumn{4}{c}{$p(\xi_{G}=j,\xi_{R}=k)$} & &
& \multicolumn{6}{c}{[with difference (and its SD)]} \\
setting &(0,0)&(0,1)&(1,0)&(1,1)&OR&&
\multicolumn{2}{c}{$p(\xi_{G}=1)$} &
\multicolumn{2}{c}{$p(\xi_{R}=1)$} \\
\hline
\textbf{Overall} &  $.24$ & $.09$ & $.14$ & $.53$ & $10.6$ && $.67$ & & $.62$ & \\
\multicolumn{11}{l}{\textbf{Respondent (child) characteristics}} \\
\multicolumn{11}{l}{\emph{Age}} \\
\hspace*{1em}35 years & $.22$ & $.09$ & $.16$ & $.54$ & $9.1$ && $.70$ & & $.63$ & \\
\hspace*{1em}45 years & $.28$ & $.10$ & $.14$ & $.48$ & $9.5$ && $.62$ & $-.07^{***\phantom{}}$ $(.02)$ & $.58$ & $-.05^{**\phantom{*}}$ $(.02)$ \\
\hspace*{1em}55 years & $.31$ & $.07$ & $.14$ & $.48$ & $18.2$ && $.62$ & $-.07^{**\phantom{*}}$ $(.03)$ & $.55$ & $-.08^{*\phantom{**}}$ $(.04)$ \\
\multicolumn{11}{l}{\emph{Gender}} \\
\hspace*{1em}Female& $.23$ & $.13$ & $.14$ & $.50$ & $6.5$ && $.65$ & & $.63$ & \\
\hspace*{1em}Male& $.26$ & $.04$ & $.13$ & $.56$ & $30.0$ && $.70$ & $+.05^{***\phantom{}}$ $(.02)$ & $.60$ & $-.03^{\phantom{***}}$ $(.02)$ \\
\multicolumn{11}{l}{\emph{Partnership status}} \\
\hspace*{1em}Single& $.25$ & $.08$ & $.09$ & $.57$ & $19.7$ && $.66$ & & $.65$ & \\
\hspace*{1em}Partnered& $.24$ & $.09$ & $.15$ & $.52$ & $9.1$ && $.67$ & $+.01^{\phantom{***}}$ $(.02)$ & $.61$ & $-.05^{**\phantom{*}}$ $(.02)$ \\
\multicolumn{11}{l}{\emph{Age of youngest coresident child}}\\
\hspace*{1em}No children& $.24$ & $.08$ & $.15$ & $.52$ & $10.5$ && $.68$ & $+.05^{\phantom{***}}$ $(.03)$ & $.60$ & $-.05^{\phantom{***}}$ $(.04)$ \\
\hspace*{1em}0-1 years& $.25$ & $.12$ & $.10$ & $.53$ & $13.1$ && $.63$ & & $.65$ & \\
\hspace*{1em}2-4 years& $.19$ & $.10$ & $.13$ & $.59$ & $9.7$ && $.72$ & $+.09^{***\phantom{}}$ $(.03)$ & $.68$ & $+.03^{\phantom{***}}$ $(.04)$ \\
\hspace*{1em}5-10 years& $.24$ & $.11$ & $.04$ & $.60$ & $51.0$ && $.64$ & $+.02^{\phantom{***}}$ $(.03)$ & $.71$ & $+.06^{*\phantom{**}}$ $(.04)$ \\
\hspace*{1em}11-16 years& $.29$ & $.06$ & $.18$ & $.48$ & $13.8$ && $.65$ & $+.02^{\phantom{***}}$ $(.03)$ & $.54$ & $-.11^{***\phantom{}}$ $(.04)$ \\
\hspace*{1em}17-- years& $.26$ & $.07$ & $.14$ & $.52$ & $14.6$ && $.66$ & $+.03^{\phantom{***}}$ $(.04)$ & $.59$ & $-.06^{\phantom{***}}$ $(.05)$ \\
\multicolumn{11}{l}{\emph{Number of siblings}} \\
\hspace*{1em}No sibling& $.21$ & $.11$ & $.17$ & $.51$ & $6.3$ && $.68$ & & $.62$ & \\
\hspace*{1em}1 sibling& $.23$ & $.10$ & $.12$ & $.55$ & $10.3$ && $.67$ & $-.01^{\phantom{***}}$ $(.02)$ & $.65$ & $+.03^{\phantom{***}}$ $(.03)$ \\
\hspace*{1em}2 or more& $.26$ & $.07$ & $.14$ & $.52$ & $13.0$ && $.66$ & $-.01^{\phantom{***}}$ $(.02)$ & $.59$ & $-.03^{\phantom{***}}$ $(.03)$ \\
\multicolumn{11}{l}{\emph{Longstanding illness}}\\
\hspace*{1em}Yes& $.28$ & $.10$ & $.13$ & $.49$ & $11.4$ && $.62$ & & $.60$ & \\
\hspace*{1em}No& $.24$ & $.09$ & $.14$ & $.53$ & $10.4$ && $.67$ & $+.05^{***\phantom{}}$ $(.02)$ & $.62$ & $+.02^{\phantom{***}}$ $(.02)$ \\
\multicolumn{11}{l}{\emph{Employment status}} \\
\hspace*{1em}Not employed& $.27$ & $.09$ & $.18$ & $.46$ & $8.3$ && $.65$ & & $.55$ & \\
\hspace*{1em}Employed& $.24$ & $.09$ & $.13$ & $.55$ & $11.8$ && $.67$ & $+.03^{\phantom{***}}$ $(.02)$ & $.64$ & $+.09^{***\phantom{}}$ $(.02)$ \\
\multicolumn{11}{l}{\emph{Education}} \\
\hspace*{1em}Secondary or less& $.26$ & $.07$ & $.15$ & $.52$ & $13.4$ && $.67$
&
& $.59$ & \\
\hspace*{1em}Post-secondary& $.23$ & $.10$ & $.13$ & $.54$ & $9.6$ && $.67$
& $-.00^{\phantom{***}}$ $(.01)$
& $.64$ &
$+.05^{***\phantom{}}$ $(.02)$ \\
\multicolumn{11}{l}{\emph{Household tenure}} \\
\hspace*{1em}Own home& $.24$ & $.08$ & $.15$ & $.53$ & $10.7$ && $.68$ & & $.61$ & \\
\hspace*{1em}Renter& $.25$ & $.10$ & $.12$ & $.53$ & $11.5$ && $.64$ & $-.03^{*\phantom{**}}$ $(.02)$ & $.63$ & $+.02^{\phantom{***}}$ $(.02)$ \\
\multicolumn{11}{l}{\emph{Logarithm of household equivalised income}} \\
\hspace*{1em}25 percentile& $.25$ & $.08$ & $.14$ & $.53$ & $12.0$ && $.67$ & & $.61$ & \\
\hspace*{1em}50 percentile& $.24$ & $.09$ & $.14$ & $.53$ & $10.7$ && $.67$ & $-.01^{**\phantom{*}}$ $(.00)$ & $.62$ & $+.00^{\phantom{***}}$ $(.00)$ \\
\hspace*{1em}75 percentile& $.24$ & $.10$ & $.14$ & $.52$ & $9.6$ && $.66$ & $-.01^{**\phantom{*}}$ $(.01)$ & $.62$ & $+.00^{\phantom{***}}$ $(.01)$ \\
\multicolumn{11}{l}{\textbf{Parent characteristics}}\\
\multicolumn{11}{l}{\emph{Age of the oldest living parent}}\\
\hspace*{1em}65 years& $.29$ & $.11$ & $.04$ & $.56$ & $46.3$ && $.60$ & & $.67$ & \\
\hspace*{1em}70 years& $.25$ & $.15$ & $.08$ & $.52$ & $12.6$ && $.60$ & $.00^{\phantom{***}}$ $(.01)$ & $.67$ & $+.00^{\phantom{***}}$ $(.01)$ \\
\hspace*{1em}80 years& $.22$ & $.06$ & $.20$ & $.52$ & $9.7$ && $.72$ & $+.12^{***\phantom{}}$ $(.02)$ & $.58$ & $-.10^{***\phantom{}}$ $(.03)$ \\
\multicolumn{11}{l}{\emph{At least one parent lives alone}}\\
\hspace*{1em}Yes& $.22$ & $.05$ & $.19$ & $.55$ & $14.5$ && $.74$ & & $.59$ & \\
\hspace*{1em}No& $.27$ & $.11$ & $.09$ & $.53$ & $14.0$ && $.62$ & $-.11^{***\phantom{}}$ $(.01)$ & $.64$ & $+.05^{**\phantom{*}}$ $(.02)$ \\
\multicolumn{3}{l}{\textbf{Child-parent characteristics}} \\
\multicolumn{11}{l}{\emph{Travel time to the nearest parent}}\\
\hspace*{1em}$>$ 1 hour& $.51$ & $.07$ & $.11$ & $.31$ & $22.7$ && $.42$ & & $.38$ & \\
\hspace*{1em}$\leq$ 1 hour& $.15$ & $.10$ & $.15$ & $.61$ & $6.2$ && $.76$ & $+.34^{***\phantom{}}$ $(.02)$ & $.70$ & $+.32^{***\phantom{}}$ $(.02)$ \\
\hline
\multicolumn{13}{l}{{\footnotesize{The posterior credible interval excludes zero at
level 90\% (*), 95\% (**) or 99\% (***).}}}
\end{tabular}
\end{small}
\end{table}

\clearpage
\bibliographystyle{apacite}
\bibliography{ref}

\end{document}